\documentclass[final,twocolumn,twoside]{IEEEtran}

\usepackage{wrapfig}

\usepackage{cite}
\usepackage{amsmath,amssymb,amsfonts}
\usepackage{algorithmic}

\usepackage{textcomp}
\def\BibTeX{{\rm B\kern-.05em{\sc i\kern-.025em b}\kern-.08em
    T\kern-.1667em\lower.7ex\hbox{E}\kern-.125emX}}

\usepackage{hyperref}
\hypersetup{hidelinks=true}
\usepackage{graphicx}
\usepackage{amsmath,amssymb,theorem}
\usepackage{listings}
\usepackage{graphicx}
\usepackage{algorithm,algorithmic}
\usepackage{psfrag}
	
\usepackage{color,xspace}
\usepackage{subfigure}

\newtheorem{theorem}{Theorem}[section]
\newtheorem{lemma}[theorem]{Lemma}

\newtheorem{remark}[theorem]{Remark}
	
\newtheorem{corollary}[theorem]{Corollary}

\newtheorem{problem}[theorem]{Problem}
\newtheorem{conjecture}[theorem]{Conjecture}

\newcommand{\real}{{\mathbb{R}}}

\newcommand{\disturb}{\delta}

\renewcommand{\hat}{\widehat}

\newcommand{\oprocendsymbol}{\hbox{$\bullet$}}
\newcommand{\oprocend}{\relax\ifmmode\else\unskip\hfill\fi\oprocendsymbol}

\begin{document}
\title{Asymptotic Stability of Active Disturbance Rejection Control for Linear SISO Plants with Low Observer Gains
}

\author{James Berneburg \qquad Daigo Shishika \qquad Cameron Nowzari \thanks{This work was supported in part by the Department of the Navy, Office of Naval Research (ONR), under federal grant N00014-22-1-2207. 
The authors are
    with the Department of Electrical and Computer Engineering,
    George Mason University, Fairfax, VA 22030, USA,  {\tt\small
      \{jbernebu,dshishik,cnowzari\}}@gmu.edu}}%

\maketitle

\begin{abstract}
This paper theoretically investigates the closed-loop performance of active disturbance rejection control (ADRC) on a third-order linear plant with relative degree 3, subject to a class of exogenous disturbances. While PID control cannot be guaranteed to be capable of stabilizing such plants, ADRC offers a model-free alternative. However, many existing works on ADRC consider the observer gains to be taken arbitrarily large, in order to guarantee desired performance, such as works which consider parameterizing ADRC by bandwidth. 
This work finds that, 
for constant exogenous disturbances, arbitrary eigenvalue assignment is possible for the closed-loop system under linear ADRC, thus guaranteeing the existence of an ADRC controller for desired performance without taking any gains arbitrarily large. We also find that stabilization is possible when the exogenous disturbance is stable, and show how ADRC can recover the performance of model-based observers. We demonstrate aspects of the resulting closed-loop systems under ADRC in simulations.
\end{abstract}

\section{Introduction}


Model-based control provides rigorous theoretical results and allows a wide variety of plants, both linear and nonlinear, to be controlled provided the plant is known~\cite{ogata2010modern,khalil2002nonlinear}. 
Observer-based control allows the controller to estimate the state, with the result that the controller is much better at dealing with noise~\cite{ogata2010modern}. 
However, the assumption of a good model of the plant is restrictive~\cite{ZG:2014}. 
Considerable effort goes into finding perfect ``correct'' models, because this minimizes disturbances due to uncertainties in the model~\cite{ZG:2014}. 
This winds up being very expensive, requiring, potentially, a lot of time and effort. 

On the other hand, proportional-integral-derivative (PID) control does not require a model of the plant and is widely used in industry applications~\cite{ogata2010modern,ZG:2014, JH:2009, QGW-CH-XPY:2001, SWS-IBL:1996}. Instead of a model, it only requires the tuning of its several gain parameters, the effects of which are well understood. However, it does not work in all applications and has some limitations, such as requiring the derivative of the output and its integral term causing stability problems~\cite{ZG:2014, JH:2009, SJ-RD:1996, SWS-IBL:1996, ASR-MC:2006}. 
Additionally, PID control cannot stabilize some higher-order plants with desired performance unless it is further modified~\cite{SJ-RD:1996,paraskevopoulos1980design,QGW-CH-XPY:2001,ASR-MC:2006}. 
For example, controlling these might require access to higher order derivative terms, but those are even more difficult to estimate in the constant presence of noise, and 
modifications, such as the addition of a lead/lag compensator~\cite{ASR-MC:2006}, can deal with these problems at the cost of complexity. As a result, while model-based control can handle a great variety of plants at the cost of requiring a good model of the plant, PID can be tuned much more simply but is limited in applications. Therefore, a relevant question is: are there any alternative controllers which can deal with more complicated plants without requiring a model? 

One candidate for such a controller is active disturbance rejection control (ADRC). Several works consider ADRC as an alternative to PID control~\cite{JH:2009,ZG:2014,WHC-JY-LG-SL:2015,ZHW-HCZ-BZG-FD:2020,ZW-ZG-DL-YC-YL:2021}. ADRC is a type of model-free controller which requires very little information about the plant~\cite{ZG:2014, JH:2009}. 
It is related to the nonlinear control technique feedback linearization~\cite{YH-WX:2014}. 
In feedback linearization, a change of coordinates is used to find a canonical form for a system, where the system is represented as a chain of integrators, with all of the nonlinear aspects of the system grouped together so that they can be canceled by the input~\cite{khalil2002nonlinear}. 
However, this still requires a very good estimation of those nonlinearities, so it still requires a very good plant model. 
ADRC uses a similar principle of canceling some plant dynamics to leave a chain of integrators, but, instead of relying on a model, ADRC uses an extended state observer (ESO)~\cite{YH-WX:2014}. As a result, ADRC has been proposed to control many different kinds of plants, both linear and nonlinear, including ones of higher order~\cite{YH-WX:2014,xue2013frequency,ZLZ-BZG:2016,zheng2007stability}. 
This means that ADRC does not require a model of the plant, but is able to control some plants which PID control cannot. 

As a result, ADRC has seen success in some areas, such as in simulations and with theoretical guarantees. See~\cite{YH-WX:2014,ZHW-HCZ-BZG-FD:2020} for an overview of more theoretical results and for specific applications of ADRC, and 
see~\cite{ZW-ZG-DL-YC-YL:2021} for discussion on applying ADRC to thermal processes, such as coal power plants, and existing results on that topic. In~\cite{WHC-JY-LG-SL:2015}, specific industry applications of ADRC are listed, including use in servo motors and in some of Texas Instruments' motion control chips~\cite{instruments2014technical}. 
On the theoretical side, the work in~\cite{BZG-ZLZ:2011} and~\cite{QZ-LG-ZG:2012} provides conditions for convergence of the observer error to within specified bounds for the ESO.  
\cite{zhao2010adrc} develops an ADRC controller for a specific problem of intercepting a target, and validates it in simulations. Similarly,~\cite{SZ-ZG:2013} uses ADRC to suppress vibrations in resonant systems and validates the proposed controller in simulations. 
The work~\cite{tian2007frequency} considers ADRC in the frequency domain for second order linear plants, and uses numerical simulations to find that stability margins can experience little change, in spite of significant variation in the plant parameters. \cite{xue2013frequency} provides related theoretical results using the frequency domain, theoretically guaranteeing asymptotic stability of linear systems under ADRC with a reduced-order observer for a range of plant parameters. 
Furthermore, it examines how to achieve desired crossover frequencies and stability margins. This paper~\cite{ZG:2006} introduces the bandwidth parameterization and discusses its use for ADRC, providing a theoretical guarantee on the BIBO stability of the closed-loop system, when the derivative of the total disturbance is treated as the input. \cite{zhao2012adrc} considers the problem of stabilizing a minimum-phase plant with an uncertain relative degree, and finds that ADRC is able to do so. 
\cite{xue2013parameters} claims to provide conditions for exponential stability of a class of nonlinear plants under ADRC, subject to a limited sampling rate, for state feedback. \cite{ZLZ-BZG:2016} and \cite{zheng2007stability} provide practical stability results for applying ADRC to nonlinear systems, and~\cite{MR-QW-CD:2016} additionally considers input saturation. Similarly, \cite{WX-YH:2013} provides results on trajectory tracking with ADRC on nonlinear systems. \cite{BZG-ZLZ:2013} provides practical convergence results for nonlinear ADRC applied to linear systems. 

However, in spite of these successes, a disadvantage of ADRC is that, unlike PID control, it has so many parameters that tuning them directly is impractical. 
Instead, many works use a parameterization. Some consider a high gain parameterization for ADRC's extended state observer (ESO)~\cite{BZG-ZLZ:2011,MR-QW-CD:2016}, and many works~\cite{YH-WX:2014,WX-YH:2013,WX-WB-SY-KS-YH:2015,SZ-ZG:2013,ZW-ZG-DL-YC-YL:2021,xue2013frequency,xue2013parameters} use a special case of it known as the bandwidth parameterization~\cite{ZG:2006}. 
Additionally, the control law is sometimes parameterized by bandwidth as well~\cite{QZ-LG-ZG:2012,ZG:2006,tian2007frequency,zheng2007stability,tan2015linear}. 
Desired performance is often handled by considering convergence of the plant state to a reference trajectory~\cite{WX-YH:2013,WX-WB-SY-KS-YH:2015,xue2013frequency,xue2013parameters}, with the tracking becoming arbitrarily good as the bandwidth is increased. 
These methods allow for some types of desired performance simply by tuning a few parameters. 

However, these parameterizations have some known drawbacks. 
For example, in the context of practical convergence, arbitrarily close convergence may require the observer gains to be arbitrarily high, and a similar problem exists for tracking a desired reference trajectory to achieve desired performance. Such resulting high gain observers are not robust to noise, and will require higher sampling rates. A few works~\cite{xue2013parameters,herbst2023tuning,WX-WB-SY-KS-YH:2015} consider a limited sampling rate for bandwidth parameterized ADRC. The latter work~\cite{WX-WB-SY-KS-YH:2015} considers sensor noise and proposes a time-varying observer gain to better handle it at the cost of a more complicated controller. Additionally, trajectory tracking does not consider the cost of control effort, and one known problem of high gain observers is the peaking phenomenon, where the magnitude of the control input spikes at initialization~\cite{YH-WX:2014,freidovich2008performance,ZLZ-BZG:2016}. Some methods of handling peaking include a time-varying gain for the ESO~\cite{ZLZ-BZG:2016}, nonlinear observer gains~\cite{ZHW-HCZ-BZG-FD:2020}, saturating the input~\cite{freidovich2008performance}, and setting the input to zero initially~\cite{YH-WX:2014}. These specifically handle peaking, at the cost of additional complication in the controller. 
Similarly, the work~\cite{du2022analysis} considers a more complicated modification of ADRC to allow stabilization with lower gains.

Consequently, we have high gain parameterizations, such as the bandwidth parameterization, on one hand, which offers limited types of performance but is very simple to tune. On the other hand, we have tuning all the gains of ADRC individually, which, while is impractical to do by hand, can potentially offer much wider variety in terms of performance. Therefore, it is relevant to examine what types of performance can be achieved in general, without relying on high gains. 

Therefore, we are most interested in works which guarantee asymptotic convergence of ADRC, because such guarantees do not require arbitrarily high observer gains. However, many works which do provide asymptotic stability guarantees require restrictive assumptions. 
For example, a bounded-input bounded-output stability result is provided in~\cite{ZG:2006}, 
but this can only guarantee asymptotic stability when the plant is an ideal chain of integrators. Similarly,~\cite{xue2013parameters} guarantees the existence of an exponentially stabilizing bandwidth-parameterized ADRC controller for a given sampling rate, but it assumes state feedback rather than output feedback for the controller. 
\cite{zheng2007stability} only guarantees asymptotic stability if the observer employs the plant dynamics. 
An exception is the work in~\cite{xue2013frequency}, which guarantees asymptotic stability for a range of plant parameters for a given bandwidth for linear plants. 
However, this work still does not guarantee desired performance without arbitrarily high gains. In order to guarantee desired stability margins, it still requires the bandwidth to be taken arbitrarily high. 
Also, the results of~\cite{WX-YH:2013} are guarantees for exponential stability of nonlinear plants under some assumptions on the plant dynamics 
and no external disturbance, but its only guarantees for desired performance are bounds on the tracking error of a reference trajectory which depend on high bandwidths. 
In~\cite{jin2023notions}, ADRC is found to stabilize certain plants, without requiring arbitrarily high bandwidths, but it only allows for some particular frequency-domain performance characteristics. 
Additionally, the work in~\cite{freidovich2008performance}, has some overlap with ADRC although it is not labeled as such. One can obtain an ADRC controller from the control form presented in that paper, and the paper is able to guarantee exponential stability, assuming input saturation, under some seemingly less-conservative assumptions on the nonlinear plant. However, it, likewise, relies on the tracking of a trajectory for desired performance, potentially requiring an arbitrarily fast observer. 

On the other hand, while not directly examining ADRC's performance, a similar problem of characterizing ADRC's capabilities is considered in~\cite{zhou2018analysis}, which similarly finds fault with the limitations of the bandwidth parameterization. 
It is able to guarantee that ADRC can realize transfer functions in a broad class, and the work~\cite{zhou2020implementation} modifies ADRC so that it can realize still more transfer functions. 
However, this does not provide insight into how to tune ADRC directly for any performance metric at all, instead requiring the intermediate step of finding a controller transfer function which provides the desired performance and can be realized by ADRC. 

Therefore, the goal of this work is to provide not only asymptotic convergence guarantees, but also guarantees on desired performance, without requiring arbitrarily fast observers, for a class of plants. 
The main contribution of this paper is to show ADRC can asymptotically stabilize a class of third-order linear plants for desired performance with low gains by providing a method to assign the gains, provided that the plant parameters are known. More specifically, we show that arbitrary eigenvalue assignment is possible with proper choice of ADRC's gains for third-order linear plants with relative degree three, subject to a class of disturbances. We additionally show how ADRC can recover the performance of model-based observers by applying the results of~\cite{zhou2018analysis}, and conclude by suggesting relevant research directions. 

\section{Notation}
The Euclidean norm of a vector $v \in \real^n$ is denoted by $||v||$. An n-dimensional column vector with every entry equal to $1$ (or $0$) is denoted by $\mathbf{1}_n$ (or $\mathbf{0}_n$). 
Given a vector~$v \in \real^N$, we denote by $\operatorname{diag}(v)$ the $N\times N$ diagonal matrix with the entries of $v$ along its diagonal. 
For a function $f$, its $n$th derivative with respect to time is denoted by $f^{(n)}$, for $n \in \mathbb{Z}_{\geq 0}$. 
The minimum eigenvalue of a square matrix $A$ is given by $\operatorname{eigmin}(A)$ and its maximum eigenvalue is given by $\operatorname{eigmax}(A)$. For a square matrix $A \in \mathbb{R}^{N \times N}$, let $\text{eig}(A)$ denote the set of the eigenvalues of $A$. 
For a complex number $p\in \mathbb{C}$, let $\text{Im}\lbrace p \rbrace$ denote its imaginary part and $\text{Re}\lbrace p \rbrace$ denote its real part. 

\section{Canonical Form and ADRC}
We consider linear system of dimension $N=3$ and relative degree $\rho = 3$, perturbed by some unknown disturbance $\disturb \in \mathbb{R}^3$, which can be written as
\begin{align}\label{eq:mainLinearPlant}
\dot{z} &= Az + Bu + \disturb,\nonumber\\
y &= Cz
\end{align}
where $z \in \mathbb{R}^3$ is the plant state and $u \in \mathbb{R}$ is the input, $y \in \mathbb{R}$ is the output, and $A\in \mathbb{R}^{3\times 3}$, $B\in \mathbb{R}^{3\times 1}$, and $C\in \mathbb{R}^{1\times 3}$ are the plant parameter matrices. For analysis purposes, we wish to transform this system into the canonical form of feedback linearization~\cite{YH-WX:2014}. Starting from~\eqref{eq:mainLinearPlant}, we define a change of coordinates by taking repeated derivatives of the output, until the result depends explicitly on the input:
\begin{align*}
x_1 &= y = Cz\\
x_2 &= \dot{y} = CAz + CBu + C\disturb\\
x_3 &= y^{(2)} = CA^2z + CABu + CA\disturb + C\dot{\disturb},\\
\dot{x}_3 &= CA^3z + CA^2Bu + CA^2\disturb + CA\dot{\disturb} + C\disturb^{(2)}.
\end{align*}
Note that, by the definition of relative degree, we have $CB=0$, $CAB=0$, and $CA^2B \neq 0$. Note also that we have made some assumption on the differentiability of $\disturb(t)$ with respect to time, and we assume that $\disturb(t)$ and its derivatives are independent of $z$ and $u$. Now $x$ can be written as
\begin{align}
    x &= \mathcal{O}z + \left[\begin{matrix}
0\\
C\\
CA
\end{matrix}\right]\disturb + \left[\begin{matrix}
0\\
0\\
C
\end{matrix}\right] \dot{\disturb},
\end{align}
where $\mathcal{O}$ is the observability matrix. From~\cite{khalil2002nonlinear}, we know that the transformation from $z$ to $x$ will be invertible, which in this case implies that $\mathcal{O}$ is invertible and so
\begin{align*}
z &= \mathcal{O}^{-1}x - \mathcal{O}^{-1}\left[\begin{matrix}
0\\
C\\
CA
\end{matrix}\right]\disturb - \mathcal{O}^{-1}\left[\begin{matrix}
0\\
0\\
C
\end{matrix}\right] \dot{\disturb}.
\end{align*}
Now we can write
\begin{align*}
\dot{x}_1 &= x_2\\
\dot{x}_2 &= x_3\\
\dot{x}_3 &=  CA^3 \mathcal{O}^{-1}x + d(t) + CA^2Bu, 
\end{align*}
where we've defined a disturbance term
\begin{align*}
d(t) &\triangleq \left(CA^2 - CA^3\mathcal{O}^{-1}\left[\begin{matrix}
0\\
C\\
CA
\end{matrix}\right] \right) \disturb \\
&+\left(CA - CA^3\mathcal{O}^{-1}\left[\begin{matrix}
0\\
0\\
C
\end{matrix}\right] \right) \dot{\disturb} + C\disturb^{(2)} .
\end{align*}
Defining the plant parameters 
$\mathbf{a} = \left[ \begin{matrix}
a_1 & a_2 & a_3
\end{matrix}\right] = CA^3\mathcal{O}^{-1}$ and 
$b = CA^2B$ for convenience, we can write
\begin{align}\label{eq:canonical3Dplant}
\dot{x} &= \left[ \begin{matrix}
0 & 1 & 0\\
0 & 0 & 1\\
a_1 & a_2 & a_3
\end{matrix} \right] x + \left[\begin{matrix}
0\\ 0\\ 1 
\end{matrix}\right] d(t) + \left[\begin{matrix}
0\\ 0\\ b 
\end{matrix}\right] u \nonumber\\
y &= \left[ \begin{matrix}
1 & 0 & 0
\end{matrix}\right] x.
\end{align}
This is the canonical form for~\eqref{eq:mainLinearPlant}, and this derivation shows how it and the disturbance $d(t)$ are related to the original plant. 
\begin{remark}[PID Control]\label{rm:PIDisInsufficient}
{\rm 
    Note that PID control cannot stabilize the system~\eqref{eq:canonical3Dplant} in general, in the sense that there exist plant parameters $a_1,a_2,a_3$ such that there do not exist stabilizing PID gains. See the appendix for more specific mathematics and a proof of this claim. 
    } \oprocend
\end{remark}
Instead, to control this, we consider ADRC with a linear extended state observer (ESO) as follows
\begin{align}\label{eq:ADRC3Dcontrol}
u &= -\frac{1}{\hat{b}}(K\hat{x} + \hat{d})\nonumber\\
\left[\begin{matrix}
\dot{\hat{x}}\\
\dot{\hat{d}}
\end{matrix} \right] &= 
\left[ \begin{matrix}
0 & 1 & 0 & 0\\
0 & 0 & 1 & 0\\
0 & 0 & 0 & 1\\
0 & 0 & 0 & 0
\end{matrix}\right] \left[ \begin{matrix}
\hat{x} \\
\hat{d}
\end{matrix} \right]
 + \left[\begin{matrix}
0\\
0\\
\hat{b}\\
0
\end{matrix} \right] u - G \left(\hat{x}_1-y\right),
\end{align}
where $G \in \mathbb{R}^4$ is the observer gain, and $K \in \mathbb{R}^3$ is the controller gain, and $\hat{b}$ is an estimate of the input's coefficient. From the canonical form of the plant, it's apparent that $\hat{x}$ can act as an estimate of $x$, while $\hat{d}$ accounts for the extra terms in $\dot{x}_3$. 



\begin{remark}[Total Disturbance]\label{rm:totalDisturbance}
{\rm 
    The total disturbance, which is often considered in ADRC literature, is in this case not $d(t)$, but rather
    \begin{align}
        \mathbf{a}x + d(t) + (b-\hat{b})u
    \end{align}
    However, we consider the feedback portion of the disturbance, which depends on the plant state $x$ and input $u$, separately from exogenous disturbance $d(t)$, for analysis purposes. 
    } \oprocend
\end{remark}

\subsection{Closed-Loop System}
Now that we have a system in canonical form and an observer and control law for it, we can investigate its closed-loop properties. 
For the sake of analysis, we consider a constant disturbance term $d$, so that $\dot{d}=0$ and we can write the closed-loop system as linear. 
Defining $\tilde{x} \triangleq \hat{x} - x$ and $\tilde{d} \triangleq \frac{b}{\hat{b}} \hat{d} - d$, we calculate $\dot{\Tilde{x}} = \dot{\hat{x}} - \dot{x}$ and $\dot{\tilde{d}} = \frac{b}{\hat{b}}\dot{\hat{d}} - \dot{d} = \frac{b}{\hat{b}}\dot{\hat{d}}$ and derive
\begin{align}\label{eq:ADRC3dCL}
\left[\begin{matrix}
\dot{x}\\
\dot{\tilde{x}}\\
\dot{\tilde{d}}
\end{matrix}\right] &= A_{CL}\left[\begin{matrix}
x\\
\tilde{x}\\
\tilde{d}
\end{matrix}\right],
\end{align}
for some matrix $A_{CL} \in \mathbb{R}^{7 \times 7}$. For the sake of space, we do not write out the matrix $A_{CL}$ in general, but in the case where $\hat{b} = b$, we have 
\begin{align}\label{eq:ADRC3dCLmatrix}
&A_{CL} =\\
& \left[\begin{matrix}
0 & 1 & 0 & 0 & 0 & 0 & 0\\
0 & 0 & 1 & 0 & 0 & 0 & 0\\
a_1-k_1 & a_2-k_2 & a_3-k_3 & -k_1 & -k_2 & -k_3 & -1\\
0 & 0 & 0 & -g_1 & 1 & 0 & 0\\
0 & 0 & 0 & -g_2 & 0 & 1 & 0\\
-a_1 & -a_2 & -a_3 & -g_3 & 0 & 0 & 1\\
0  & 0 & 0 & -g_{4} & 0 & 0 & 0
\end{matrix}\right] \nonumber
.
\end{align}
Note that some nice cancellation has occurred because we have assumed that $\hat{b}=b$. 
Now, unless $a_1=a_2=a_3=0$, this is not upper block triangular, and the separation principle does not hold for designing the gains $K$ and $G$. 
Note that this occurs because of the mismatch between the plant dynamics and the observer dynamics, where the observer assumes that $a_1=a_2=a_3=0$. 
As a result, we cannot use conventional methods to design the control gains $K$ or the observer gains $G$. This motivates the use of parameterizations to ensure desired performance without needing to know the plant parameters $\mathbf{a}$, such as the bandwidth parameterization. However, motivated by concerns such as noise sensitivity and limited sampling rates, we wish to find methods of tuning for desired performance without requiring arbitrarily high observer gains. 

\section{Eigenvalue Assignment with known plant parameters}

In order to see what performance is possible when tuning ADRC's gains $K$ and $G$, we consider this as a pole placement problem and assume that $\mathbf{a}$ is known to the designer. Although $\mathbf{a}$ will not be known in practice, this will show us when it is possible to design ADRC with an ESO as in~\eqref{eq:ADRC3Dcontrol} to place the poles at a desired location. For example, we would like to show that it is possible to ensure stability and desired performance without taking the observer gains arbitrarily high. Note that this still leaves open the problem of how to find these gains, since, in practice, $\mathbf{a}$ will not be known to the designer. 

We assume that we have a set of desired eigenvalues $-p_1^*,\dots,-p_7^*$ for the closed-loop system matrix~\eqref{eq:ADRC3dCLmatrix} and we wish to determine when it is possible to achieve them through the proper choice of $K$ and $G$.  
Note that these desired eigenvalues correspond to a desired characteristic polynomial 
\begin{align}\label{eq:desiredCharPoly}
    Q^*(s) &\triangleq s^7 +q_6^*s^6 +q_5^*s^5 +q_4^*s^4 \notag\\ &+q_3^*s^3 +q_2^*s^2 +q_1^*s +q_0,
\end{align}
which we assume has real coefficients $q_0^*,\dots,q_6^*\in \mathbb{R}$. 
If we take the roots of~\eqref{eq:desiredCharPoly} to be $-p_1^*,\dots,-p_7^*$, the eigenvalue assignment problem is equivalent to matching the characteristic polynomial of~\eqref{eq:ADRC3dCLmatrix} to~\eqref{eq:desiredCharPoly}. 
We formalize this problem as follows. 

\begin{problem}[ADRC Eigenvalue Assignment]\label{prob:ADRCpolePlacementInitial}
For a given $\mathbf{a}$, find $K \in \mathbb{R}^3$ and $G \in \mathbb{R}^4$ such that the characteristic polynomial of the closed-loop system matrix~\eqref{eq:ADRC3dCLmatrix} is equal to a given desired polynomial as defined in~\eqref{eq:desiredCharPoly}, $Q^*(s)$ 
which can be written as
\begin{align}
    \text{det}\left(A_{CL} - I_7s \right) &= Q^*(s).
\end{align}
\end{problem}

We wish to determine under what circumstances Problem~\ref{prob:ADRCpolePlacementInitial} has a solution. 
However, before attempting to solve this problem, we review the nominal eigenvalue assignment problem, where $a_1=a_2=a_3=0$. 

This can similarly be formalized as follows. 

\begin{problem}[Nominal Eigenvalue Assignment]\label{prob:ADRCpolePlacementInitialNom}
Find $K \in \mathbb{R}^3$ and $G \in \mathbb{R}^4$ such that the characteristic polynomial of the nominal closed-loop system matrix~\eqref{eq:ADRC3dCLmatrix}, with $\mathbf{a}=0$, is equal to a given desired polynomial $Q^*(s)$ with real coefficients. 
\end{problem}

Note that~\eqref{eq:ADRC3dCLmatrix} is upper block triangular when $\mathbf{a}=0$, so we can apply the separation principle, and determine whether~\ref{prob:ADRCpolePlacementInitialNom} can be solved by examining the diagonal blocks.

\begin{lemma}[Nominal Eigenvalue Assignment]\label{lm:nomEigvalAssign}
For all $Q^*(s)$ as defined in~\eqref{eq:desiredCharPoly},
there exist $K \in \mathbb{R}^3$ and $G \in \mathbb{R}^4$ to solve Problem~\ref{prob:ADRCpolePlacementInitialNom}. 
\end{lemma}

\begin{proof}
Because~\eqref{eq:ADRC3dCLmatrix} is upper block triangular when $\mathbf{a}=0$, its eigenvalues will be the eigenvalues of the diagonal blocks
\begin{align*}
\left[\begin{matrix}
0 & 1 & 0 \\
0 & 0 & 1 \\
-k_1 & -k_2 & -k_3 
\end{matrix}\right],
\left[\begin{matrix}
-g_1 & 1 & 0 & 0\\
-g_2 & 0 & 1 & 0\\
-g_3 & 0 & 0 & 1\\
-g_4 & 0 & 0 & 0
\end{matrix}\right] ,
\end{align*}
where the first matrix corresponds to the plant and the second to the observer error. To see if the eigenvalues of these can be placed arbitrarily with proper choice of $K,G$ we turn to observability and controllability. Defining 
\begin{align}\label{eq:observerPlant}
\hat{A}_{p} &\triangleq \left[\begin{matrix}
0 & 1 & 0\\
0 & 0 & 1\\
0 & 0 & 0
\end{matrix}\right],
&\hat{B}_{p} &\triangleq \left[\begin{matrix}
0\\
0\\
1
\end{matrix}\right],
\end{align}
we have that $(\hat{A}_{p}$, $\hat{B}_{p})$ is controllable, which implies that the eigenvalues of the plant block can be placed arbitrarily with proper choice of $K\in \mathbb{R}^3$~\cite{ogata2010modern}. 
Since the observer uses an extended state, we must consider a different matrix for observability and so we similarly define \begin{align}\label{eq:observerError}
\hat{A}_{e} &\triangleq \left[ \begin{matrix}
0 & 1 & 0 & 0\\
0 & 0 & 1 & 0\\
0 & 0 & 0 & 1\\
0 & 0 & 0 & 0
\end{matrix}\right], 
&\hat{C}_{e} &\triangleq \left[ \begin{matrix}
1 & 0 & 0 & 0
\end{matrix}\right],
\end{align}
and we have that $(\hat{A}_{e}$, $\hat{C}_{e})$ is observable, so the eigenvalues of the observer error block can similarly be placed arbitrarily with proper choice of $G \in \mathbb{R}^4$~\cite{ogata2010modern}. Therefore, Problem~\ref{prob:ADRCpolePlacementInitialNom} can be solved for any $Q^*(s)$ with real coefficients. 
\end{proof}


Before moving on, note the characteristic polynomial of this nominal closed-loop matrix is
\begin{align*}
&\text{det}\left(A_{CL} - I_7s \right) =\\
& s^7 + \hat{q}_6s^6 + \hat{q}_5s^5 + \hat{q}_4s^4 + \hat{q}_3s^3 + \hat{q}_2s^2 + \hat{q}_1s + \hat{q}_0,
\end{align*}
when $\mathbf{a} = \mathbf{0}_3$, and by matching coefficients, we have the following system of equations 
\begin{align}\label{eq:nomCharPolyCoefficients}
q_6^* &= \hat{q}_6 \triangleq g_1 + k_3 \notag\\
q_5^* &= \hat{q}_5 \triangleq g_2 + k_2 + g_1k_3 \notag\\
q_4^* &= \hat{q}_4 \triangleq g_3 + k_1 + g_1k_2 + g_2k_3\notag\\
q_3^* &= \hat{q}_3 \triangleq g_4 + g_1k_1 + g_2k_2 + g_3k_3\notag\\
q_2^* &= \hat{q}_2 \triangleq g_2k_1 + g_3k_2 + g_4k_3\notag\\
q_1^* &= \hat{q}_1 \triangleq g_3k_1 + g_4k_2\notag\\
q_0^* &= \hat{q}_0 \triangleq g_4k_1 .
\end{align}
Therefore, the nominal eigenvalue placement problem can be written as finding $K\in \mathbb{R}^3,G\in \mathbb{R}^4$ such that~\eqref{eq:nomCharPolyCoefficients} is satisfied. 
Therefore, this system of equations has a solution for $K,G$ iff~\ref{prob:ADRCpolePlacementInitialNom} has a solution, and from~\ref{lm:nomEigvalAssign} we know that the latter has a solution. 
As a result, this system of equations has at least one solution for $K\in \mathbb{R}^3,G\in \mathbb{R}^4$, for any real $q_0^*,\dots,q_6^*$. 

We now return to considering the general closed-loop matrix $A_{CL}$ with $\mathbf{a} \neq \mathbf{0}_3^T$ and Problem~\ref{prob:ADRCpolePlacementInitial}. Note that we cannot use the same technique to show that it is possible to place its eigenvalues at desired locations, with proper choice of $K\in \mathbb{R}^3,G\in \mathbb{R}^4$, because it is not upper block triangular in general. 
However, we can similarly find that the characteristic polynomial of the closed-loop matrix is
\begin{align*}
&\text{det}\left(A_{CL} - I_7s \right) = \\
&s^7 + q_6s^6 + q_5s^5 +q_4s^4 +q_3s^3 +q_2s^2 +q_1s +q_0,
\end{align*}
and by matching the coefficients with those of the desired closed-loop characteristic polynomial~\eqref{eq:desiredCharPoly}, we have the following system of equations
\begin{align}\label{eq:mainPolePlaceCoeffMatch}
    q_6^* &= q_6 \triangleq g_1 + k_3 - a_3 \notag\\
    q_5^* &= q_5 \triangleq g_2 + k_2 + g_1k_3 - a_2 - a_3(g_1 + k_3) \notag\\
    q_4^* &= q_4 \triangleq g_3 + k_1 + g_1k_2 + g_2k_3 - a_1 \notag\\
    &- a_2(g_1 + k_3) - a_3(k_2 + g_2 + g_1k_3) \notag\\
    q_3^* &= q_3 \triangleq \frac{b}{\hat{b}}\left(g_4 + g_1k_1 + g_2k_2 + g_3k_3\right) - a_1(g_1 + k_3) \notag\\
    &- a_2(g_2 + k_2 + g_1k_3) - a_3(g_3 + k_1 + g_1k_2 + g_2k_3) \notag\\
    q_2^* &= q_2 \triangleq \frac{b}{\hat{b}}\left(g_2k_1 + g_3k_2 + g_4k_3\right) - a_1(g_2 + k_2 + g_1k_3) \notag\\
    &- a_2(g_3 + k_1 + g_1k_2 + g_2k_3) \notag\\
    q_1^* &= q_1 \triangleq \frac{b}{\hat{b}}\left(g_3k_1 + g_4k_2\right) - a_1(g_1k_2 + g_3 + k_1 + g_2k_3) \notag\\
    q_0^* &= q_0 \triangleq \frac{b}{\hat{b}}g_4k_1 .
\end{align}
The eigenvalue assignment problem becomes one of finding $K\in \mathbb{R}^3,G\in \mathbb{R}^4$ to satisfy~\eqref{eq:mainPolePlaceCoeffMatch}. 
Having set up the system of equations that must be solved, we are ready to present our main result. 

\begin{theorem}[Arbitrary Eigenvalue Assignment]\label{th:eigvalAssignment}
Given the plant dynamics in~\eqref{eq:canonical3Dplant} and the input defined by~\eqref{eq:ADRC3Dcontrol}, with $\dot{d}=0$
, for any desired characteristic polynomial $Q^*(s)$~\eqref{eq:desiredCharPoly}, corresponding to desired closed-loop eigenvalues, 
there exist $K \in \mathbb{R}^3$ and $G \in \mathbb{R}^4$ such that characteristic polynomial of the closed-loop system matrix~$A_{CL}$
is equal to the desired one. Mathematically, this can be written as
\begin{align*}
    \text{det}\left(A_{CL} - I_7s \right) = Q^*(s) .
\end{align*}
\end{theorem}

\begin{proof}

With the way the terms are grouped in~\eqref{eq:mainPolePlaceCoeffMatch}, we can see that this can be written more simply by substituting in $\hat{q}_0,\dots,\hat{q}_6$, as defined in~\eqref{eq:nomCharPolyCoefficients}, so that 
\begin{align*}
q_6^* &= \hat{q}_6 - a_3\\
q_5^* &= \hat{q}_5 - a_2 - a_3\hat{q}_6\\
q_4^* &= \hat{q}_4 - a_1 - a_2\hat{q}_6 - a_3\hat{q}_5\\
q_3^* &= \frac{b}{\hat{b}}\hat{q}_3 - a_1\hat{q}_6 - a_2\hat{q}_5 - a_3\hat{q}_4 \\
q_2^* &= \frac{b}{\hat{b}}\hat{q}_2 - a_1\hat{q}_5 - a_2\hat{q}_4 \\
q_1^* &= \frac{b}{\hat{b}}\hat{q}_1 - a_1\hat{q}_4 \\
q_0^* &= \frac{b}{\hat{b}}\hat{q}_0 .
\end{align*}
This is now a linear system of equations in $\hat{q}_0,\dots,\hat{q}_6$, and after some algebra and substitutions, defining $\hat{q}_i^*$ 
to simplify notation, we have
\begin{align*}
\hat{q}_6 &= \hat{q}_6^* \triangleq q_6^* + a_3\\
\hat{q}_5 &= \hat{q}_5^* \triangleq q_5^* + a_2 + a_3\hat{q}_6^*\\
\hat{q}_4 &= \hat{q}_4^* \triangleq q_4^* + a_1 + a_2\hat{q}_6^* + a_3\hat{q}_5^*\\
\hat{q}_3 &= \hat{q}_3^* \triangleq \frac{\hat{b}}{b}\left(q_3^* + a_1\hat{q}_6^* + a_2\hat{q}_5^* + a_3\hat{q}_4^* \right)\\
\hat{q}_2 &= \hat{q}_2^* \triangleq \frac{\hat{b}}{b}\left(q_2^* + a_1\hat{q}_5^* + a_2\hat{q}_4^* \right)\\
\hat{q}_1 &= \hat{q}_1^* \triangleq \frac{\hat{b}}{b}\left(q_1^* + a_1\hat{q}_4^*\right)\\
\hat{q}_0 &= \hat{q}_0^* \triangleq \frac{\hat{b}}{b}q_0^* .
\end{align*}

Therefore, this is simply a problem of matching the \emph{nominal} closed-loop characteristic polynomial to some other polynomial, which is determined by the desired closed-loop eigenvalues and $\mathbf{a}$. 

As a result, the eigenvalue assignment problem for the closed-loop matrix $A_{CL}$, where $\mathbf{a} \neq \mathbf{0}_3^T$ in general, is equivalent to an eigenvalue assignment problem for the nominal closed-loop matrix, where $\mathbf{a} = \mathbf{0}_3^T$, with the desired characteristic polynomial given by
\begin{align}\label{eq:desiredNomCharPoly}
s^7 + \hat{q}_6^*s^6 + \hat{q}_5^*s^5 + \hat{q}_4^*s^4 + \hat{q}_3^*s^3 + \hat{q}_2^*s^2 + \hat{q}_1^*s + \hat{q}_0^*.
\end{align}
Therefore, we have written the eigenvalue assignment problem~\ref{prob:ADRCpolePlacementInitial} with the desired characteristic polynomial given by~\eqref{eq:desiredCharPoly}
as a nominal eigenvalue assignment problem~\ref{prob:ADRCpolePlacementInitialNom} with the desired characteristic polynomial given by~\eqref{eq:desiredNomCharPoly}. 
The polynomial~\eqref{eq:desiredNomCharPoly} has real coefficients, so by Lemma~\ref{lm:nomEigvalAssign}, this nominal eigenvalue assignment problem can be solved. 
Therefore, Problem~\ref{prob:ADRCpolePlacementInitial} can be solved for any $Q^*(s)$ with real coefficients. 

\end{proof}


\begin{remark}[Relationship to~\cite{zhou2018analysis}]
{\rm 
    Note that we could have, alternatively, used the results of~\cite{zhou2018analysis} to prove Theorem~\ref{th:eigvalAssignment}. In order to do so, we would consider matching the transfer function of an ADRC controller with a third order ESO, by proper choice of $K$ and $G$, to a desired controller transfer function, which places the closed-loop eigenvalues in desired locations. We would then use their results on ADRC's ability to realize strictly proper controller transfer functions with pure integrators in place of Lemma~\ref{lm:nomEigvalAssign} to guarantee that a solution exists for $K$ and $G$. We do not do so because writing our controller in the frequency domain does not result in simpler presentation, and it does not result in additional insight because~\cite{zhou2018analysis} converts their problem in state space to guarantee a solution. 
    } \oprocend
\end{remark}

To solve the eigenvalue assignment problem~\ref{prob:ADRCpolePlacementInitial} for~\eqref{eq:ADRC3dCLmatrix}, we can find the roots of~\eqref{eq:desiredNomCharPoly}, pick two appropriate poles to assign using $K$ and~\eqref{eq:observerPlant}, and assign the remaining three using $G$ and~\eqref{eq:observerError}. 
These gains can be found by conventional methods, such as by using Ackermann's formula. 
Since this system of equations has a real solution for $K\in\mathbb{R}^3,G\in \mathbb{R}^4$, for a desired characteristic polynomial~\eqref{eq:desiredCharPoly} of the closed-loop system matrix~\eqref{eq:ADRC3dCLmatrix}, the desired closed-loop eigenvalues can be achieved by ADRC. 

\begin{remark}[Non-Uniqueness of Solutions]\label{rm:additionalFreedom}
{\rm
    Note that, after having specified the desired closed-loop poles, there is still an addition degree of freedom, in general. Although~\eqref{eq:desiredNomCharPoly} is uniquely determined by the plant parameters and the desired closed-loop eigenvalues, there are \emph{multiple} nominal problems which we can formulate to find gains to achieve that characteristic polynomial. Specifically, we must choose some roots of~\eqref{eq:desiredNomCharPoly} to assign with $K$ and assign the rest with $G$ for the nominal problem, and that choice is not unique. At this point, it is not clear what effect different choices of roots of~\eqref{eq:desiredNomCharPoly} to assign with $K$ will have on the performance, because because the closed-loop eigenvalues will be the same in any case. 
    } \oprocend
\end{remark}

\begin{remark}[Does the separation principle hold?]\label{rm:separationQuestion}
{\rm
Because we have converted a problem where the separation apparently doesn't hold, to one where it's used in the solution, there's a question that may arise.  Does the separation principle apply here? The answer depends on what exactly is meant by the separation principle applying. While it can now be used to solve for $K$ and $G$ such that the eigenvalues are placed in the desired locations, we do not have any guarantee that any eigenvalues only correspond to the observer error $\hat{x} - x$; instead, each eigenvalue may correspond to both the plant state and observer error. 
This suggests that the separation principle does not hold in the sense of closed-loop performance, for the state and observer error as defined. 
} \oprocend
\end{remark}

We will briefly examine both Remark~\ref{rm:additionalFreedom} and Remark~\ref{rm:separationQuestion} with examples in Section~\ref{sec:sims}.

While the results of this section indicate that ADRC can provide desired performance without requiring arbitrarily high gains, this is not a practical method of choosing the gains $G$ and $K$, because the plant parameters $\mathbf{a}$ will not be known in practice. 
On the other hand, this seems to suggest that, for our particular system, the specific value of $\hat{b}$ is immaterial and need not be a good estimate of $b$, if the closed-loop eigenvalues are our main performance concern. Note that the value of $\hat{b}$ may still have an effect on what states those eigenvalues correspond to. However, without knowing how those are affected or knowing the exact value of $b$, one may simply choose $\hat{b}$ to be some convenient value, such as setting $\hat{b} = \text{sgn}(b)$. 

\subsection{Plants of Arbitrary Relative Degree}

Here, we briefly consider the case where the plant's relative degree can be arbitrary, rather than three. Mathematically, consider a plant with equal order and relative degree ($N=\rho \in \mathbb{Z}_{>0}$, $x\in \mathbb{R}^N$), which with a slight abuse of notation is  
\begin{align}\label{eq:canonical3DplantN}
    \dot{x} &= \left[ \begin{matrix}
        \mathbf{0}_{\rho-1} & & \mathbf{I}_{\rho-1}\\
        & \mathbf{a} & 
    \end{matrix}\right] x + \left[\begin{matrix}
        \mathbf{0}_{\rho-1}\\
        b
    \end{matrix}\right] d
    + \left[\begin{matrix}
        \mathbf{0}_{\rho-1}\\
        b
    \end{matrix}\right] u \notag\\
    y &= \left[ \begin{matrix}
        1 & \mathbf{0}_{\rho-1}
    \end{matrix}\right] x,
\end{align}
where, with a slight abuse of notation, $\mathbf{a} \in \mathbb{R}^\rho$ and $b \neq 0$ are the plant parameters and $d \in \mathbb{R}$ is a constant disturbance. In this case, the controller with an ESO of appropriate order is, again with a slight abuse of notation,  
\begin{align}\label{eq:ADRC3DcontrolN}
u &= -\frac{1}{\hat{b}}(K\hat{x} + \hat{d})\nonumber\\
\left[\begin{matrix}
\dot{\hat{x}}\\
\dot{\hat{d}}
\end{matrix} \right] &= 
\left[ \begin{matrix}
\mathbf{0}_\rho & & \mathbf{I}_\rho\\
& \mathbf{0}^T_{\rho+1} &
\end{matrix}\right] \left[ \begin{matrix}
\hat{x} \\
\hat{d}
\end{matrix} \right]
 + \left[\begin{matrix}
\mathbf{0}_{\rho-1}\\
\hat{b}\\
0
\end{matrix} \right] u - G \left(\hat{x}_1-y\right),
\end{align}
where now $\hat{x} \in \mathbb{R}^\rho$, $K \in \mathbb{R}^\rho$, $G \in \mathbb{R}^{\rho+1}$. As before, we can write the closed-loop system as a linear one, so that
\begin{align}
\left[\begin{matrix}
\dot{x}\\
\dot{\tilde{x}}\\
\dot{\tilde{d}}
\end{matrix}\right] &= A_{CL}\left[\begin{matrix}
x\\
\tilde{x}\\
\tilde{d}
\end{matrix}\right],    
\end{align}
where, as before, $\tilde{x} = \hat{x} - x$ and $\tilde{d} = \frac{b}{\hat{b}}\hat{d} - d$, but now $A_{CL} \in \mathbb{R}^{(2\rho + 1) \times (2\rho + 1)}$. Considering again that we have a desired characteristic polynomial, corresponding to some desired eigenvalues, which is, again with abuse of notation, 
\begin{align}\label{eq:desiredCharPolyN}
    Q^*(s) s^{2\rho + 1} + q^*_{2\rho}s^{2\rho} + \dots + q^*_1s + q^*_0,
\end{align}
where $q^*_0, q^*_1,\dots , q^*_{2\rho} \in \mathbb{R}$, we can offer the following conjecture. 

%

\begin{conjecture}[Plants of arbitrary relative degree]\label{cj:eigvalAssignmentArbitraryOrder}
Given the plant dynamics in~\eqref{eq:canonical3DplantN} and the input defined by~\eqref{eq:ADRC3DcontrolN}, with $\dot{d}=0$
, for any desired characteristic polynomial $Q^*(s)$~\eqref{eq:desiredCharPolyN}, corresponding to desired closed-loop eigenvalues, 
there exist $K \in \mathbb{R}^\rho$ and $G \in \mathbb{R}^{\rho+1}$ such that characteristic polynomial of the closed-loop system matrix~$A_{CL}$
is equal to the desired one. Mathematically, this can be written as
\begin{align*}
    \text{det}\left(A_{CL} - \mathbf{I}_{2\rho+1}s \right) = Q^*(s) .
\end{align*}
\end{conjecture}

Although proving this conjecture is beyond the scope of this current work, we do not foresee any obstacles to applying the same methodology that was used for Theorem~\ref{th:eigvalAssignment} to plants with arbitrary relative degree, provided it is equal to the plant order. 

\section{Stable Time-Varying Disturbance}\label{sec:TVdisturbance}

Up to this point, we have assumed that the external disturbance $d(t)$ is constant. Here, we investigate the effect of a time-varying disturbance, which we assume is generated by a stable dynamical system. Specifically, let
\begin{align}\label{eq:TVdisturb}
    d(t) &= d_{ss} + \gamma(t) 
\end{align}
where $d_{ss} \in \mathbb{R}$ is the constant steady state portion of the disturbance, and $\gamma(t) \in \mathbb{R}$ is the time-varying portion generated by the system
\begin{align}\label{eq:TVdisturbDynamics}
    \dot{\chi} &= f_d(\chi)\nonumber\\
    \gamma(t) &= C_d\chi,
\end{align}
where $\chi \in \mathbb{R}^M$ is the state of the disturbance's system and $f_d(\cdot)$ is a locally Lipschitz function mapping $\mathbb{R}^M$ to $\mathbb{R}^M$ and $C_d \in \mathbb{R}^{1\times M}$. 

With a slight abuse of notation, we redefine $\tilde{d} \triangleq \frac{b}{\hat{b}} \hat{d} - d_{ss}$. Now we can rewrite~\eqref{eq:ADRC3dCL} with the inclusion of the time-varying portion of the disturbance as 
\begin{align}\label{eq:ADRC3dCLTVdisturb}
\left[\begin{matrix}
\dot{x}\\
\dot{\tilde{x}}\\
\dot{\tilde{d}}
\end{matrix}\right] &= A_{CL}\left[\begin{matrix}
x\\
\tilde{x}\\
\tilde{d}
\end{matrix}\right] + B_{CL}C_d\chi ,
\end{align}
where $B_{CL} = [0,0,1,0,0,-1,0]^T$ due to the presence of $\gamma(t)$ in $\dot{x}_3$ and $\dot{\tilde{x}}_3$. If both~\eqref{eq:ADRC3dCLTVdisturb} and~\eqref{eq:TVdisturbDynamics} are stable systems, then we expect the cascaded system to also be stable.

\begin{corollary}[Stable Time-Varying Disturbance]\label{co:stableTVdisturbance}
    Given the system~\eqref{eq:ADRC3dCLTVdisturb} subject to a disturbance generated by the system~\eqref{eq:TVdisturbDynamics}, if $A_{CL}$ is a stable (Hurwitz) matrix and~\eqref{eq:TVdisturbDynamics} is globally asymptotically stable to the origin, then~\eqref{eq:ADRC3dCLTVdisturb} is also globally asymptotically stable to the origin.
\end{corollary}

\begin{proof}
    Because~\eqref{eq:ADRC3dCLTVdisturb} is a stable linear system, treating $\chi$ as the input, we can claim that it is input-to-state stable. The result then follows from the properties of input-to-state stable systems, specifically Lemma 4.7 of~\cite{khalil2002nonlinear}.
\end{proof}

This result indicates that, if we design the ADRC gains $K$ and $G$ such that the closed-loop system is stable, then the system will still be stable when subjected to a class of vanishing time-varying disturbances. 
We will examine the transient effect of this disturbance in the numerical simulations section. 

\section{Recovering the Performance of Model-Based Observers with Standard ADRC}

While we have shown that ADRC can provide desired performance in the sense that its gains can be chosen to place the closed-loop eigenvalues, one may wonder how this compares to the case where a model of the plant is known and is employed in the observer, so that the separation principle can be employed in the design and analysis of the system. To examine this, we examine the input and output relationship of the controllers by looking at their transfer functions, similarly to~\cite{zhou2018analysis}. The transfer function of~\eqref{eq:ADRC3Dcontrol} is denoted by $H_{ADRC}(s)$, so that, when employing the controller, we have
\begin{align*}
U(s) = -H_{ADRC}(s)Y(s),
\end{align*}
where $U(s)$ and $Y(s)$ are the Laplace transforms of $u$ and $y$. We compare this to a controller employing the plant model, which is 
\begin{align}\label{eq:modelBasedObserver}
u &= -\frac{1}{b}(K^*\hat{x} + \hat{d})\nonumber\\
\left[\begin{matrix}
\dot{\hat{x}}\\
\dot{\hat{d}}
\end{matrix} \right] &= 
\left[ \begin{matrix}
0 & 1 & 0 & 0\\
0 & 0 & 1 & 0\\
a_1 & a_2 & a_3 & 1\\
0 & 0 & 0 & 0
\end{matrix}\right] \left[ \begin{matrix}
\hat{x} \\
\hat{d}
\end{matrix} \right]
 + \left[\begin{matrix}
0\\
0\\
b\\
0
\end{matrix} \right] u - G^* \left(\hat{x}_1-y\right),
\end{align}
where $K^* \in \mathbb{R}^3$, $G^* \in \mathbb{R}^4$ are gains chosen for desired closed-loop performance. Note that this differs from a standard observer due to the inclusion of an extended state, but one can still employ the separation principle to examine the performance of the state dynamics and the error dynamics. We use $H^*(s)$ to denote the transfer function of~\eqref{eq:modelBasedObserver}. 

Now, the question we would like to answer is: under what circumstances can the controller~\eqref{eq:ADRC3Dcontrol} have the same transfer function as~\eqref{eq:modelBasedObserver}, or, mathematically, $H_{ADRC}(s)=H^*(s)$? To answer this, we present the following result. 

\begin{theorem}[Model-Based Observer Performance]
The controller and observer~\eqref{eq:ADRC3Dcontrol} 
with transfer function $H_{ADRC}(s)$ can realize the transfer function $H^*(s)$ of the controller and observer~\eqref{eq:modelBasedObserver}, 
in the sense that, for all $b, \hat{b} \neq 0$, $\mathbf{a} \in \mathbb{R}^3$, $K^* \in \mathbb{R}^3$, and $G^* \in \mathbb{R}^4$, there exist $K$ and $G$ such that $H_{ADRC} = H^*(s)$. 
\end{theorem}

\begin{proof}

We can find that
\begin{align*}
H_{ADRC}(s) &= \frac{\hat{q}_3s^3 + \hat{q}_2s^2 + \hat{q}_1s + \hat{q}_0}{bs\left(s^3 + \hat{q}_6s^2 + \hat{q}_5s + \hat{q}_4 \right)},
\end{align*}
where $\hat{q}_6, \hat{q}_5, \dots, \hat{q}_0$ are defined in~\eqref{eq:nomCharPolyCoefficients}. Similarly, we find that 
\begin{align*}
H^*(s) &= \frac{\overline{H}^*(s)}{\underline{H}^*(s)} \\
\overline{H}^*(s) &\triangleq (g_4^* + g_1^*k_1^* + g_2^*k_2^* + g_3^*k_3^*)s^3 \\
&+ (g_2^*k_1^* + g_3^*k_2^* + g_4^*k_3^* - a_3g_4^* - a_3g_1^*k_1^* - a_3g_2^*k_2^* \\
&+ a_1g_1^*k_3^* + a_2g_2^*k_3^*)s^2\\ 
&+ (g_3^*k_1^* + g_4^*k_2^* + a_1g_1^*k_2^* - a_2g_1^*k_1^* + a_1g_2^*k_3^* \\
&- a_3g_2^*k_1^* - a_2g_4^*)s \\
&+ g_4^*(k_1^* - a_1) \\
\underline{H}^*(s) &\triangleq bs\left(
s^3 
+ (g_1^* + k_3^* - a_3)s^2 \right.\\
&+ (g_2^* + k_2^* + g_1^*k_3^* - a_3g_1^* - a_2)s \\
&\left. + (g_3^* + k_1^* + g_1^*k_2^* + g_2^*k_3^* - a_2g_1^* - a_3g_2^* - a_1) \right).
\end{align*}
Note that both $H_{ADRC}(s)$ and $H^*(s)$ have the same form, being strictly proper transfer functions of the same order with a pure integrator, so we simply have to show that there exist $K$ and $G$ such that the coefficients match. Therefore, we can apply the results of~\cite{zhou2018analysis} to say that there exist gains $K$ and $G$ such that ADRC realizes the desired transfer function $H^*(s)$, or, mathematically, $H_{ADRC}(s) = H^*(s)$. 

\end{proof}

Note that this indicates that the additional parameters $\mathbf{a}$, which come from using the plant model in the observer, do not provide any additional flexibility in terms of performance for ADRC, so they can be omitted and the gains $K$ and $G$ can be tuned instead. This may be appropriate in cases where the plant parameters are not known, because there is no need to have an accurate estimate of them to achieve the desired performance. The advantage of including such parameters in the observer may be that it simplifies tuning by allowing one to take advantage of the separation principle.

\section{Simulations}\label{sec:sims}

To demonstrate the performance of ADRC, when calculating the gains $K$ and $G$ for desired closed-loop eigenvalues, as well as the effect of the mismatch between the plant and the observer, we perform numerical simulations and observe some example trajectories. See Table~\ref{tab:controllerData} for the parameters of main controllers which we will use throughout this section. 

\begin{table}[htb]
    \centering
    \begin{tabular}{|c|c|c|c|c|}
        \hline
        Controller & $K$ & $G$ & $\hat{b}$ & eig$\left(A_{CL}\right)$\\
        \hline
        Slow & $\left[\begin{matrix}
            0.1513\\
            1.261\\
            1.059
        \end{matrix}\right]^T$ & $\left[\begin{matrix}  
            19.14\\
            161.3\\
            802.7\\
            -4876
        \end{matrix}\right]$ & $-1$ & $\begin{matrix}
            -2,-2.2,\\
            -2.4,-2.6,\\
            -2.8,-3,-3.2
        \end{matrix}$ \\
        \hline
        Fast & $\left[\begin{matrix}
            0.5365\\
            1.788\\
            1.397
        \end{matrix}\right]^T$ & $\left[\begin{matrix}
            25.80\\
            289.2\\
            1858\\
            -13983
        \end{matrix}\right]$ & $-1$ & $\begin{matrix}
            -3,-3.2,\\
            -3.4,-3.6,\\
            -3.8,-4,-4.2
        \end{matrix}$ \\
        \hline
        Bandwidth & $\left[\begin{matrix}
            1.3310\\
            3.63\\
            3.3
        \end{matrix}\right]^T$ & $\left[\begin{matrix}
            32\\
            384\\
            2048\\
            4096
        \end{matrix}\right]$ & $1$ & $\begin{matrix}
            -14.3737,\\
            -9.0600 \pm 6.6661i,\\
            -0.3253 \pm 2.8065i,\\
            -0.0778 \pm 0.6079i
        \end{matrix}$\\
        \hline
    \end{tabular}
    \caption{Parameters and resulting closed-loop eigenvalues for the plant with $a = [4,1,2]$ and $b = -1$ for different controllers}
    \label{tab:controllerData}
\end{table}

\subsection{Basic Performance}
First, we show how ADRC can stabilize an unstable plant to desired specifications. We consider an unstable plant of the form~\eqref{eq:canonical3Dplant} where $\mathbf{a} = [4,1,2]$ and $b = -1$, and we consider that the desired performance is defined by closed-loop eigenvalues of $-2,-2.2,-2.4,-2.6,-2.8,-3,-3.2$. Note that these eigenvalues were chosen to be slightly spaced out to better demonstrate our ability to place them, because attempting to place them too close makes them sensitive to numerical errors in the calculated gains. 

The initial conditions throughout the simulations are $x(0) = [1,0,0]^T$, $\hat{x} = 0$, and $\hat{d} = 0$. We take the estimate of the input coefficient to be $\hat{b}=1$, and we consider a constant disturbance of $d(t)=1$. We choose $K = [0.1513, 1.2608, 1.0586]$ and $G = [19.1414, 161.2754, 802.6627$, $ -4876.5604]^T$ to place the closed-loop eigenvalues in the desired locations. We will refer to this as the slow controller, for reasons which will become apparent later in the section. 
Note that one of the elements of $G$ is negative because of the mismatch between the signs of $\hat{b}$ and $b$. Also note that, besides placing the eigenvalues, we have attempted to choose the gains so that $K$ will be relatively small compared to $G$. The resulting trajectories are shown in Figure~\ref{fig:basicPerformance} (a). This will be used as a nominal case to compare against. 

We then consider the same setup, but with a time-varying disturbance $d(t)$. We consider that the time-varying portion of the disturbance is generated by the following linear system
\begin{align*}
    \dot{\chi} &= \left[\begin{matrix}
        0 & 1\\
        -1 & -0.7
    \end{matrix} \right]\chi \\
    d(t) &= \left[\begin{matrix}
        1 & 0
    \end{matrix}\right]\chi + 1,
\end{align*}
with initial conditions $\chi = [1,1]^T$. This system was chosen because it is stable, but slower than the rest of the system. The resulting trajectories are shown in Figure~\ref{fig:basicPerformance} (b). While this is still apparently stable, we can see that, compared to the case without a disturbance shown in Figure~\ref{fig:basicPerformance} (a), this appears to converge more slowly, indicating that the slow disturbance has an adverse effect on the transient response. 

\begin{figure*}
\subfigure[]{\includegraphics[width=.5\linewidth]{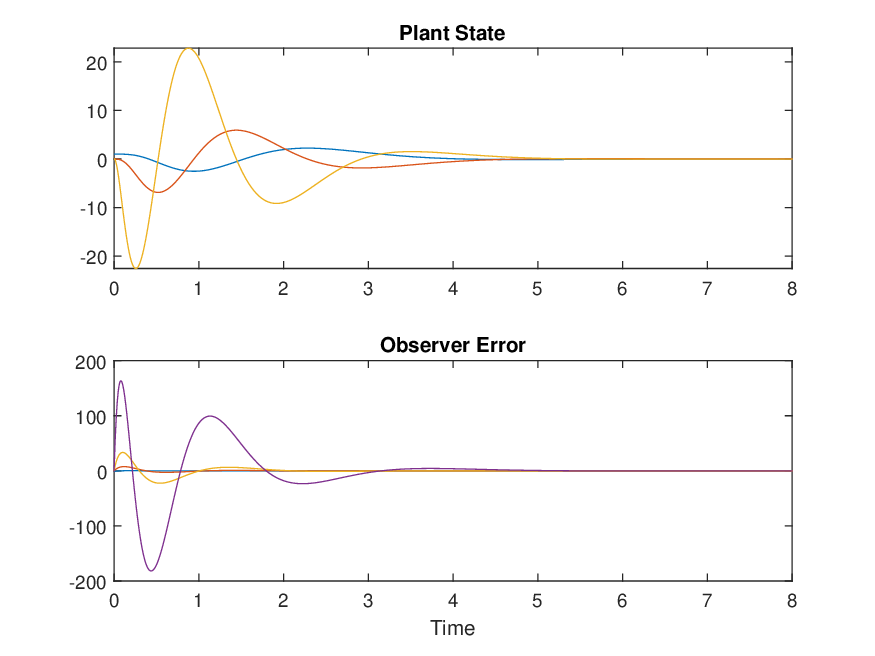}
} \hfill
\subfigure[]{\includegraphics[width=.5\linewidth]{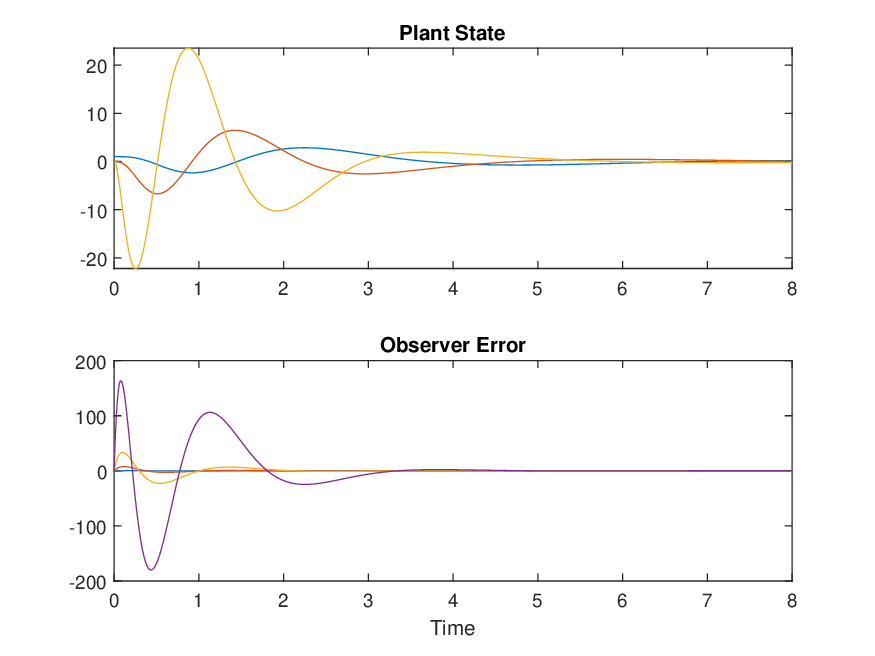}
} \hfill
\caption{This shows the performance of an ADRC controller on an unstable plant with the closed-loop eigenvalues placed at $-2$, $-2.2$, $-2.4$, $-2.6$, $-2.8$, $-3$, $-3.2$. (a) shows the trajectory under a constant disturbance, while (b) shows the trajectory under a slow-but-stable time-varying disturbance. In both cases, we can see that the system is stable, but that the time-varying disturbance results in slower convergence. } \label{fig:basicPerformance}
\end{figure*}

\subsection{Different Closed-Loop Systems with the Same Eigenvalues}

Here, we investigate the subjects of Remarks~\ref{rm:separationQuestion} and~\ref{rm:additionalFreedom} by leaving the closed-loop eigenvalues fixed, changing other parameters, and looking at example trajectories. 
First, we compare the performance on the unstable plant in the previous section to the performance on an idealized nominal plant where $\mathbf{a}=\mathbf{0}_3^T$ and $b = \hat{b}$, to demonstrate Remark~\ref{rm:separationQuestion}. We consider the same closed-loop eigenvalues of $-2,-2.2,-2.4,-2.6,-2.8,-3,-3.2$, as well as the same initial conditions and estimate of the input coefficient as in the previous section. The disturbance is constant with $d(t) = 1$.
We choose $K = [10.56, 14.48, 6.6]$ and $G = [11.6,50.36,96.976,69.888]^T$ to place the closed-loop eigenvalues in the desired locations. 
The resulting trajectories are shown in Figure~\ref{fig:idealizedPlant}. 

We can see that, in both cases, the plant state trajectories seem to converge at about the same rate. The observer error, on the other hand, converges more slowly for the unstable plant in Figure~\ref{fig:basicPerformance} (a) compared to the nominal plant in Figure~\ref{fig:idealizedPlant}. This matches our expectation that the observer error may be waiting on the plant states to converge, rather than converging on its own, while the plant may still converge at the same rate. Additionally, note that the input has a higher peak at the start in Figure~\ref{fig:basicPerformance} (a) compared to Figure~\ref{fig:idealizedPlant}, due to the unstable plant dynamics which are not accounted for by the observer. 

This shows how the unknown plant parameters, which cause a mismatch between the plant and the observer, still have some effect on the system, despite the gains being chosen such that the closed-loop eigenvalues are in the same locations. 

Next, we consider a different choice of $K$ and $G$ which lead to the same eigenvalues of $-2$, $-2.2$, $-2.4$, $-2.6$, $-2.8$, $-3$, $-3.2$ for the unstable plant, which will help investigate Remark~\ref{rm:additionalFreedom}. We choose $K = [1538.2, 232.01, 22.312]$ and $G = [-2.1117, -2.0954, -3.8457,$ $ -0.4798]^T$. While in the previous section we chose $K$ to be small and $G$ to be large, here we have done the reverse. Comparing Figure~\ref{fig:basicPerformance} (a) and Figure~\ref{fig:altKandG}, we can see that the state trajectories look very similar, and that the change seems to have mainly affected the observer error trajectories. This may suggest that the additional degree of freedom provided by the non-uniqueness of $K$ and $G$ for a given plant and given closed-loop poles may not have a significant affect on performance. 
There is not necessarily a meaningful distinction between having a relatively fast observer and having a relatively fast controller. 

\begin{figure}
\includegraphics[width=.9\linewidth]{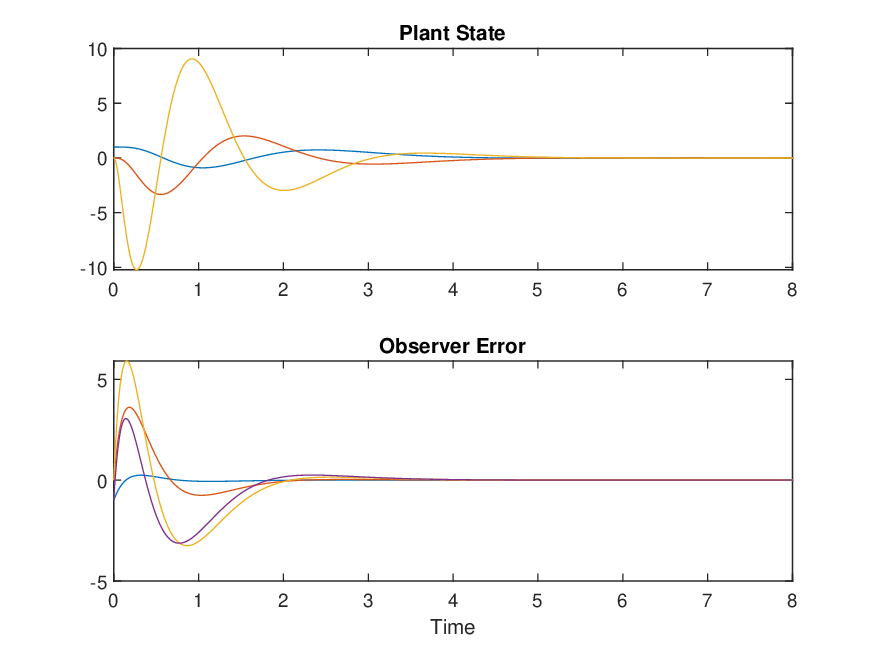}
\caption{This shows the performance of an ADRC controller on an idealized plant with the closed-loop eigenvalues placed at $-2$, $-2.2$, $-2.4$, $-2.6$, $-2.8$, $-3$, $-3.2$. Compared to the unstable plant, we can see that the observer errors converge much faster, because in this case they do not have to wait for the plant state to also converge. } \label{fig:idealizedPlant}
\end{figure}

\begin{figure}
\includegraphics[width=.9\linewidth]{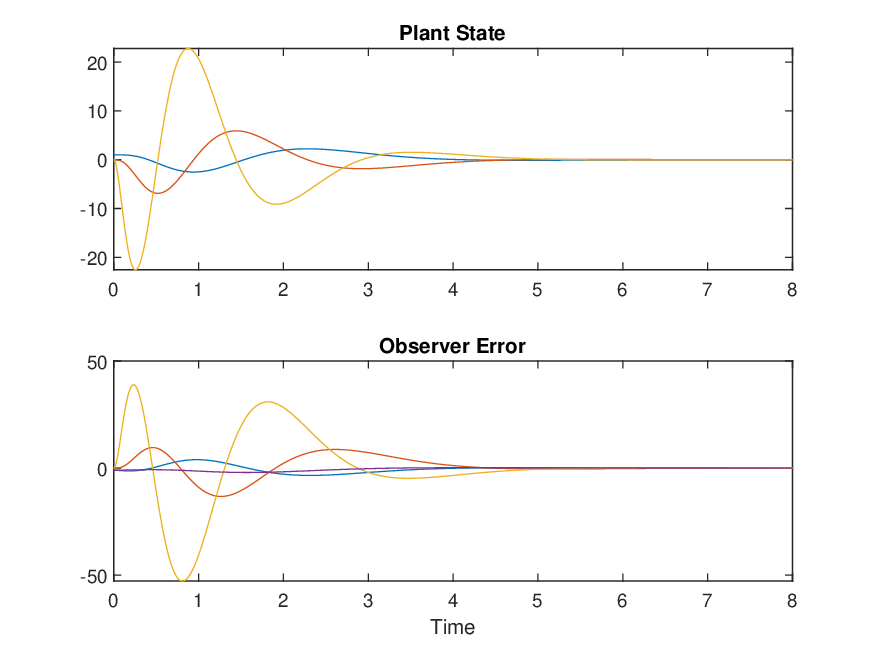}
\caption{This shows the performance of an ADRC controller on the unstable plant with the closed-loop eigenvalues placed at $-2$, $-2.2$, $-2.4$, $-2.6$, $-2.8$, $-3$, $-3.2$, with a ``fast'' controller and a ``slow'' observer. This includes the same constant disturbance of $d(t)=1$. Compared to the trajectories shown in~\ref{fig:basicPerformance}, which used a ``slow'' controller and a ``fast'' observer, we can see that the plant state trajectories are almost unchanged, while the observer error trajectories are quite different. } \label{fig:altKandG}
\end{figure}

\subsection{Time-Varying Marginally Stable Disturbance}

Here, we consider that the time-varying portion of the disturbance does not vanish, as we assumed previously. Instead, we assume that it is generated by a marginally stable linear system which only changes slowly. In particular, the time-varying portion of the disturbance is generated by the following linear system
\begin{align*}
    \dot{\chi} &= \left[\begin{matrix}
        0 & \frac{1}{2}\\
        -\frac{1}{2} & 0
    \end{matrix} \right]\chi \\
    d(t) &= \left[\begin{matrix}
        1 & 0
    \end{matrix}\right]\chi + 1,
\end{align*}
with initial conditions $\chi = [1,0]^T$. This generates a sinusoidal disturbance with a bias of $1$ and an amplitude of $1$. 
We first consider the same slow controller as before, and the resulting trajectories are shown in Figure~\ref{fig:marginalDisturbance} (a). Additionally, we compare to a faster controller where the closed-loop poles are placed at $-3$, $-3.2$, $-3.4$, $-3.6$, $-3.8$, $-4$, $-4.2$ by gains of $K = [0.5365, 1.7878, 1.3966]$ and $G = [25.8034,289.1742,1857.5406,-13983.2560]^T$. The generated trajectories are shown in Figure~\ref{fig:marginalDisturbance} (b). 

We can see that the faster controller does a better job of suppressing the disturbance's effect on the output at steady state, because Figure~\ref{fig:marginalDisturbance} (a) shows a higher magnitude in the output towards the end of the simulation than Figure~\ref{fig:marginalDisturbance} (b). However, as the next subsection will demonstrate, this comes at a cost. 

\begin{figure*}
\subfigure[]{\includegraphics[width=.5\linewidth]{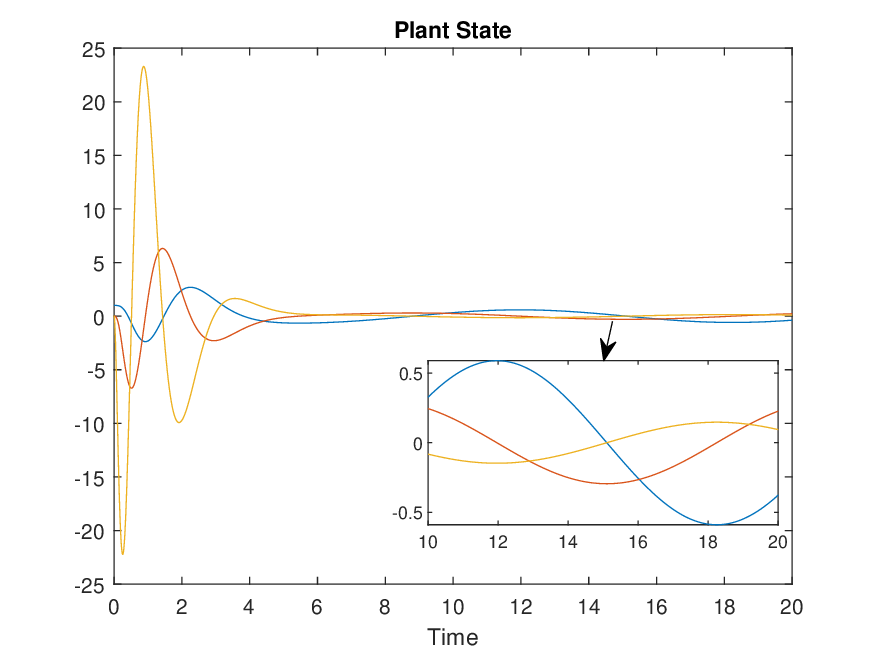}
} \hfill
\subfigure[]{
\includegraphics[width=0.5\linewidth]{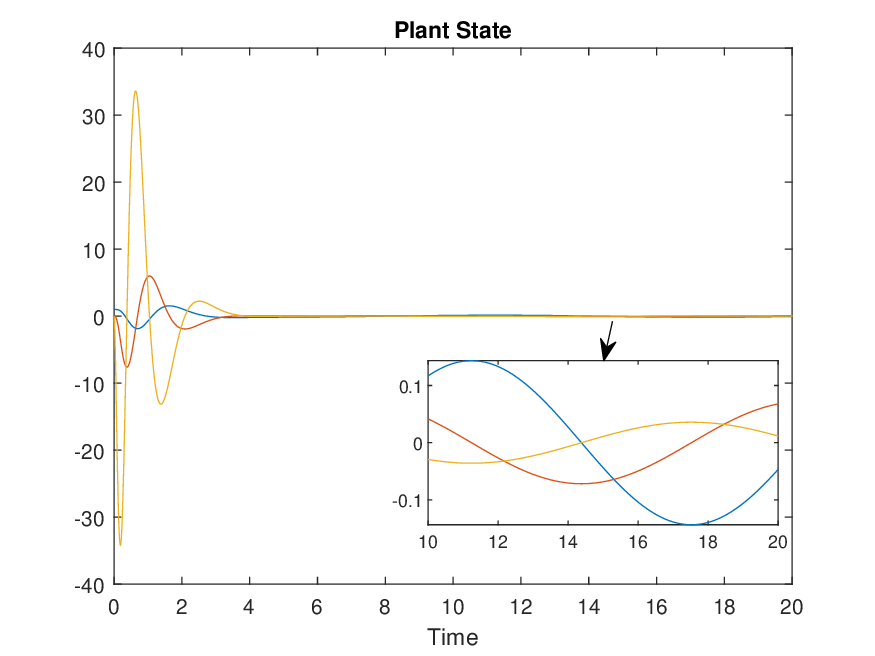}
}%
\caption{This shows the performance of ADRC on the unstable plant subjected to a non-vanishing, slowly time-varying disturbance. In (a), we consider a slower controller, with the closed-loop eigenvalues placed at $-2$, $-2.2$, $-2.4$, $-2.6$, $-2.8$, $-3$, $-3.2$. In (b), we consider a faster controller, with the closed-loop eigenvalues placed at $-3$, $-3.2$, $-3.4$, $-3.6$, $-3.8$, $-4$, $-4.2$. While both controllers attenuate the disturbance, the faster controller keeps the output closer to zero. } \label{fig:marginalDisturbance}
\end{figure*}

\subsection{Demonstration of Robustness}

Here, we show how ADRC can be used to stabilize systems such that the result is robust to noise and sampling rate. 
We consider the same unstable plant and initial conditions, with the same constant disturbance $d(t)=1$. We consider a sampling period of $0.01$, where the input is held constant between samples and the observer is only updated at those sampling times. Additionally, we consider that the plant output is corrupted by noise, so that the observer gets $\hat{y} = y + w$, where $w$ is Gaussian white noise with variance $0.0001$. We show the performance of both the slower controller and the faster controller in Figure~\ref{fig:noiseAndSamplingSlow} and Figure~\ref{fig:noiseAndSamplingFast}, respectively. 

Although both controllers are able to bring the state close to zero and keep it there, the slower one seems to be less affected by the noise, and the faster one continues to have a large input in steady state. 
On the other hand, the faster controller, with higher gains, results in inputs with higher magnitudes at the start and is more vulnerable to noise, while it is better able to handle the non-vanishing, time-varying disturbance in the previous subsection. 

\begin{figure*}
\subfigure[]{\includegraphics[width=.5\linewidth]{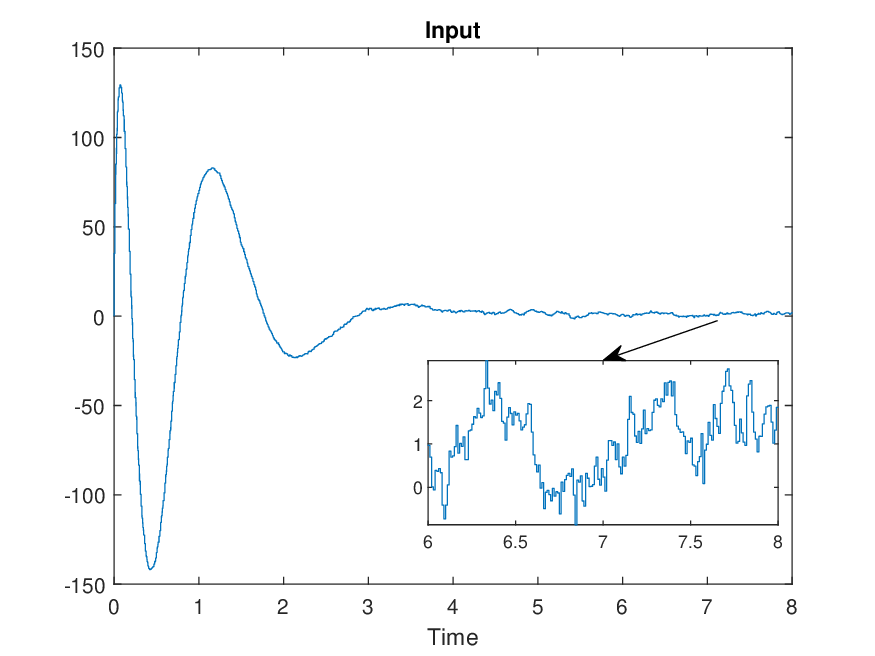}
} \hfill
\subfigure[]{\includegraphics[width=.5\linewidth]{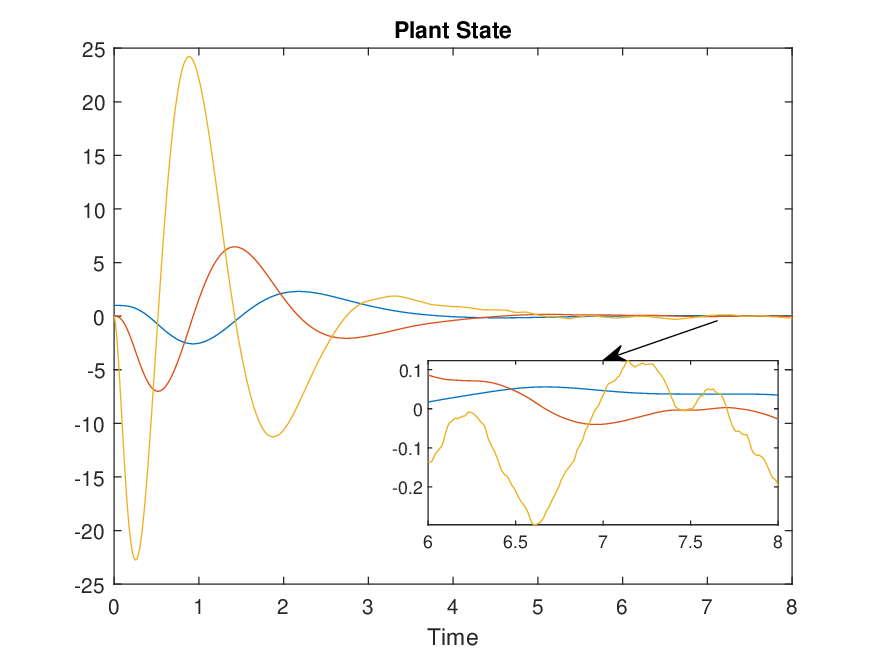}
} \hfill
\caption{This shows the performance of an ADRC controller on the unstable plant with a constant disturbance, subject to noise and a finite sampling rate. 
The input and state trajectories are shown for the slower controller, with closed-loop eigenvalues of $-2$, $-2.2$, $-2.4$, $-2.6$, $-2.8$, $-3$, $-3.2$. 
Note that, in the absence of noise, the input should asymptotically approach $1$ to counteract the constant disturbance. Compared to the faster controller in Figure~\ref{fig:noiseAndSamplingFast}, the noise is less amplified at the input, and the input's peak is smaller. 
} \label{fig:noiseAndSamplingSlow}
\end{figure*}

\begin{figure*}
\subfigure[]{\includegraphics[width=.5\linewidth]{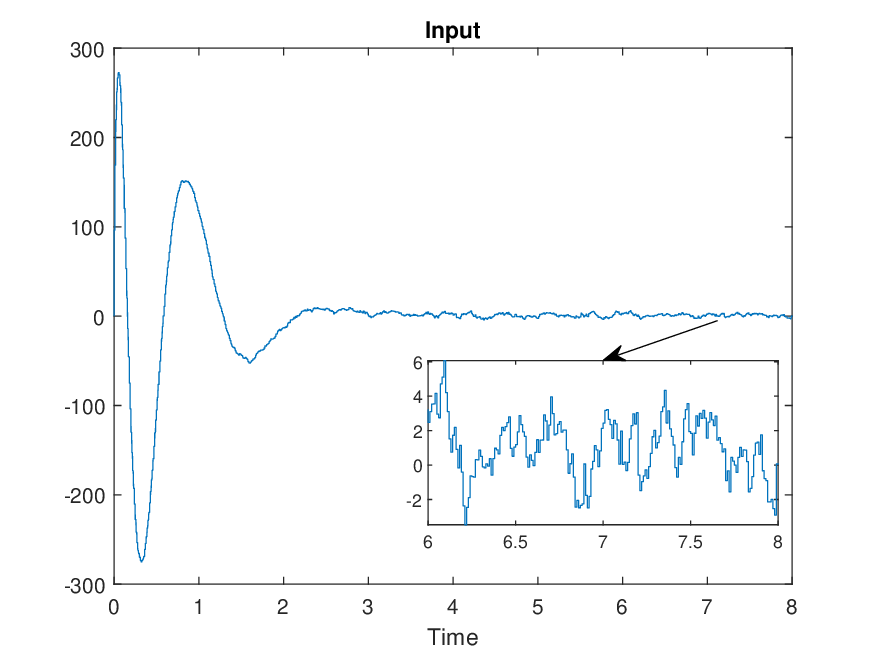}
} \hfill
\subfigure[]{\includegraphics[width=.5\linewidth]{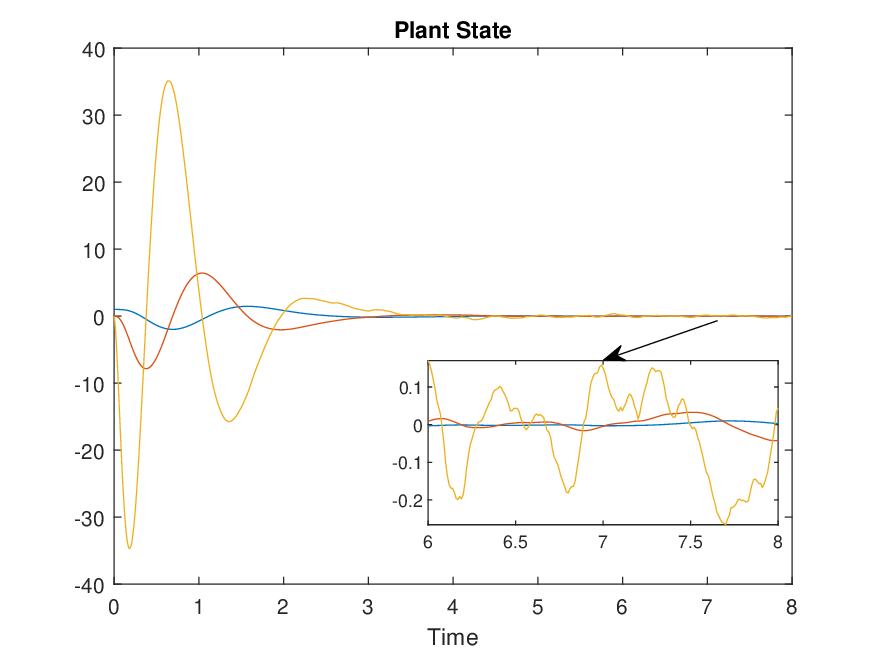}
} \hfill
\caption{This shows the performance of the faster ADRC controller, with closed-loop eigenvalues of $-3$, $-3.2$, $-3.4$, $-3.6$, $-3.8$, $-4$, $-4.2$, on the unstable plant with a constant disturbance, subject to noise and a finite sampling rate. 
The input and state trajectories for the faster controller are shown. Note that, in the absence of noise, the input should asymptotically approach $1$ to counteract the constant disturbance. Compared to the slower controller in Figure~\ref{fig:noiseAndSamplingSlow}, the noise is more amplified at the input, and the input's peak is larger. 
  } \label{fig:noiseAndSamplingFast}
\end{figure*}

\subsection{Comparison to Bandwidth Parameterization}

Here, we compare the performance of these controllers with one tuned using the bandwidth parameterization. Note that this is not an apples-to-apples comparison, because the way we calculated gains for the slower and faster ADRC controllers relied on knowledge of the plant parameters. 
To help quantify performance for tuning, we introduce the cost metric
\begin{align}\label{eq:simCost}
    \mathcal{C} = \mathcal{C}_y + \lambda \mathcal{C}_u,
\end{align}
where $\lambda > 0$ is a design parameter and 
\begin{align*}
    \mathcal{C}_y &\triangleq \int_0^\infty y^2 dt, \quad
    \mathcal{C}_u \triangleq \int_0^\infty u^2 dt.
\end{align*}

We select $\lambda = 0.1$ for our performance metric and check the cost for an individual simulation trajectory, with the same initial condition and plant as before but with $d(t)=0$. The costs shown here were calculated using a simulation time length of $30$ seconds, to approximate~\eqref{eq:simCost}. Then, we consider two parameters to tune: a controller bandwidth $\omega_c$ and an observer bandwidth $\omega_o$, where the former defines the location of the eigenvalues placed by $K$ and the latter define those placed by $G$, in the nominal problem. 
Note that some works consider a separate bandwidth for the controller in this fashion~\cite{SZ-ZG:2013,QZ-LG-ZG:2012,ZG:2006,tian2007frequency,zheng2007stability,tan2015linear}. We adjust each of these, $\omega_c$ in increments of $0.1$ and $\omega_o$ in increments of $1$, to minimize the cost~\eqref{eq:simCost}. We obtain $\omega_c = 1.1$, so $K=[1.3310, 3.63, 3.3]$, and $\omega_o=8$, so $G=[32, 384, 2048, 4096]^T$. This gives us closed-loop eigenvalues of $-14.3737$, $-9.0600 \pm 6.6661i$, $-0.3253 \pm 2.8065i$, $-0.0778 \pm 0.6079i$ and a cost of $\mathcal{C} = 1294.9$. 
For comparison, note that our slow controller achieved a lower cost of $\mathcal{C} = 987.2546$, although it has not been tuned specifically for that metric. The faster controller has a higher cost of $2801.5$. Note that although the cost metric is weighted to penalize slow convergence more than control effort, it seems as though, for all the controllers, the cost seems to be primarily determined by the size of the input peak at the start, leading to this metric preferring slower, but stable, controllers. 


We show the simulated trajectories in Figure~\ref{fig:bandwidth}. First, we show the basic performance of this controller as a baseline trajectory. Then, we show a trajectory with the same non-vanishing, slowly time-varying disturbance. Finally, we show a trajectory when the output is corrupted by the same noise as before, with the same limited sampling rate as before. 

 The controller tuned with the bandwidth parameterization seems to not only result in a higher cost than the slower controller, but also does not attenuate the marginally stable disturbance as well and appears to be more vulnerable to noise. While it does have a lower cost than the faster controller while having similar input effort, in terms of peak and steady-state range, under noise, it does not seem to handle the marginally stable disturbance as well as the faster controller, nor does the output converge as quickly. This provides an example of how, when minimizing input effort is a priority, the bandwidth parameterization may not provide desirable performance. 

\begin{figure*}
\subfigure[]{\includegraphics[width=.5\linewidth]{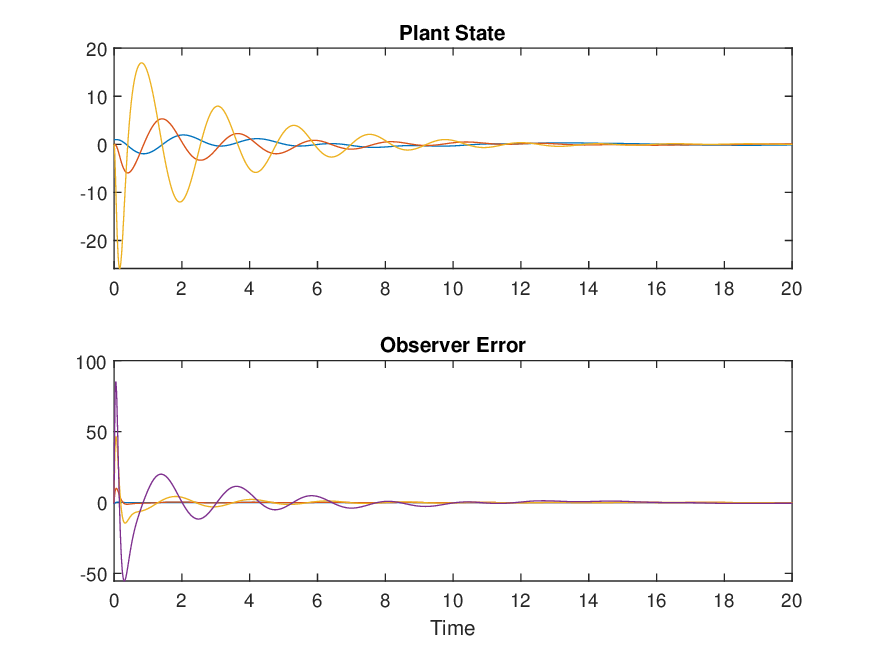}
} \hfill
\subfigure[]{\includegraphics[width=.5\linewidth]{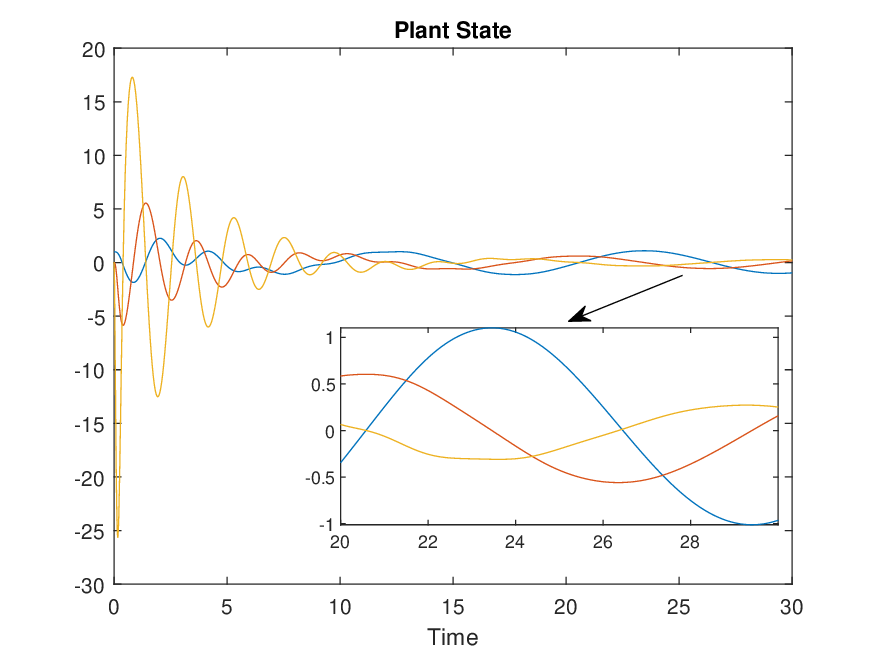}
} \hfill
\subfigure[]{\includegraphics[width=.5\linewidth]{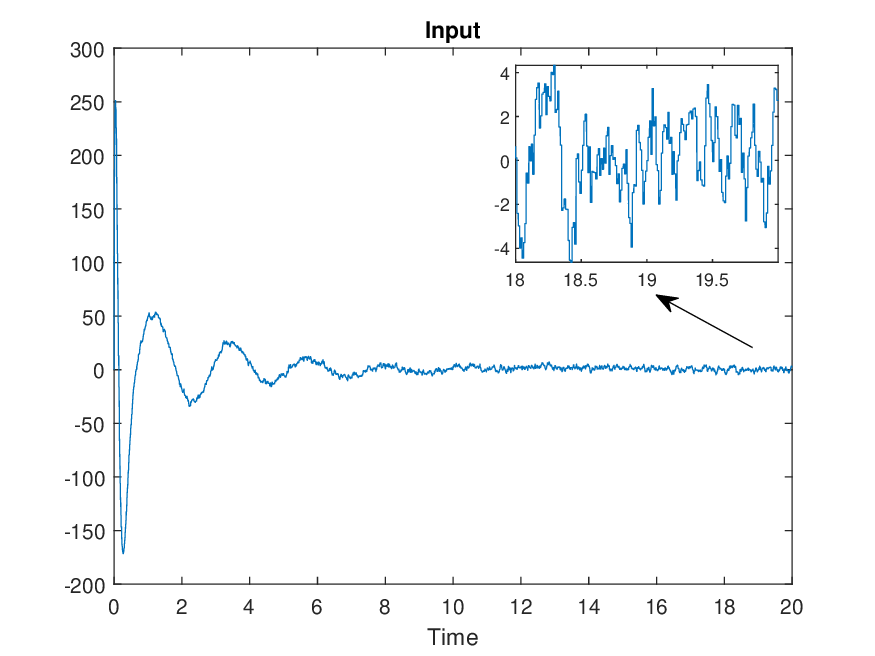}
} \hfill
\subfigure[]{\includegraphics[width=.5\linewidth]{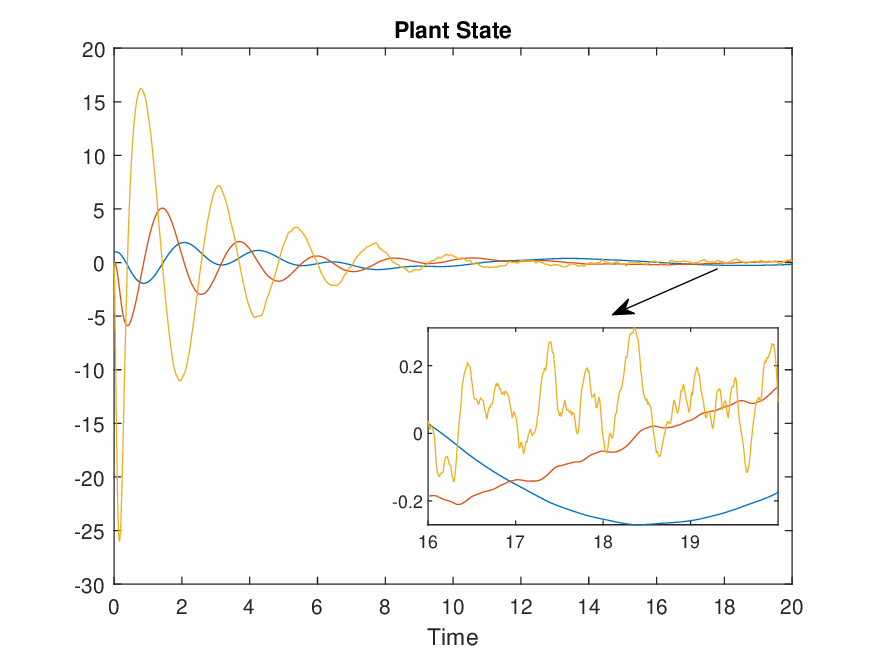}
} \hfill
\caption{This shows the performance of an ADRC controller which was tuned with the bandwidth parameterization for lower control effort on the unstable plant. (a) shows the basic performance under a constant disturbance. (b) shows a trajectory when the system is subjected to a non-vanishing, slowly time-varying disturbance. 
(c) and (d) show the input and state trajectories, respectively, when the system is subjected to output noise and a limited sampling rate. 
} \label{fig:bandwidth}
\end{figure*}

\section{Research Directions}

Here, we discuss potential future research directions which we believe are important. 

Thus far, we have considered that $d(t)$ is exogenous and does not depend in any way upon the plant state $x$. A relevant future research topic is to consider what happens when $d(t)$ has dynamics which depends on $x$. Specifically, we could consider that the disturbance takes the form~\eqref{eq:TVdisturb}, but where $\gamma(t)$ is instead generated by the following system
\begin{align}\label{eq:TVdisturbDynamicsWithInput}
    \dot{\chi} &= f_d(\chi,x) \nonumber\\
    \gamma(t) &= C_d\chi,
\end{align}
where, with some abuse of notation, $f_d(\cdot) \in \mathbb{R}^M$ now depends on both the disturbance state dynamics $\chi$ and the plant dynamics $x$. Note that this could allow us to model plants where $N> \rho$, by considering some of the original plant dynamics as part of the disturbance dynamics, at least under some assumptions on that original plant. 


One method of guaranteeing stability in such a case would be to consider that the system which generates the disturbance~\eqref{eq:TVdisturbDynamicsWithInput} is stable, similarly to Section~\ref{sec:TVdisturbance}. However, guaranteeing stability of the overall system is not straightforward in this case. In this case, because of the dependence of the disturbance system~\eqref{eq:TVdisturbDynamicsWithInput} on the plant state $x$, the plant system and the disturbance system are in a feedback loop. 
One could then apply the small-gain theorem to find sufficient conditions for stability of the overall system~\cite{khalil2002nonlinear,ogata2010modern}.  

Note, however, that such a result may be conservative and it may be possible to control systems where~\eqref{eq:TVdisturbDynamicsWithInput} is not stable. Because the disturbance $d(t)$ depends on the input $u$, at least indirectly through the plant state $x$, it may be possible to stabilize the disturbance system through proper choice of $u$. Finding what conditions are required for ADRC to do so, and what modifications could allow ADRC to do so when those conditions are not met, is another potential area of research. 

A natural extension of this work is to develop another parameterization for ADRC, which would serve as an alternative to the bandwidth parameterization. The goal would be to allow practitioners to \emph{tune} ADRC for the desired performance promised by the results of this paper, with greater flexibility than what is offered by the bandwidth parameterization, at the cost of more parameters to tune. 

One potential modest option would be to consider a linear high-gain parameterization for the observer, where the observer eigenvalues are not all placed in the same location. More specifically, consider a linear 4th order ESO, as in~\eqref{eq:ADRC3Dcontrol}, 
with a more general high-gain parameterization, such that $g_i = \alpha_i/\epsilon^i$, where $\alpha_i \in \mathbb{R}$ is a parameter to be described shortly. This is similar to the parameterization used in~\cite{freidovich2008performance} and to one mentioned in~\cite{ZHW-HCZ-BZG-FD:2020}. The parameters $\alpha_i$ are chosen such that the matrix 
\begin{align}\label{eq:HGOinitialPoles}
    A_\alpha = \left[ \begin{matrix}
        -\alpha_1 & 1 & 0 & 0\\
        -\alpha_2 & 0 & 1 & 0\\
        -\alpha_3 & 0 & 0 & 1\\
        -\alpha_4 & 0 & 0 & 0
    \end{matrix}\right]
\end{align}
is Hurwitz. Note that the observer error dynamics, in the nominal case, will then have eigenvalues of $\text{eig}(A_\alpha)/\epsilon$. The bandwidth parameterization is a special case of this, where the eigenvalues of~\eqref{eq:HGOinitialPoles} are all $-1$ and the bandwidth is $\omega_o = 1/\epsilon$. 
However, it may be beneficial to consider the eigenvalues of \eqref{eq:HGOinitialPoles} as design parameters and provide analysis and guidance as to how they should be chosen. 
The goal would be to preserve the bandwidth parameterization's trade-off between observer speed and disturbance rejection/performance, while being able to provide better performance in other respects, such as noise tolerance. 

Another direction, which is not mutually exclusive with the preceding, is to design the ADRC gains so that the transfer function of the nominal system has small $\mathcal{L}_2$ gain. To motivate this, mathematically, we can write our closed-loop system~\eqref{eq:ADRC3dCL} as a feedback between the nominal system and a disturbance system. Specifically, we can consider
\begin{align}
    \left[\begin{matrix}
\dot{x}\\
\dot{\tilde{x}}\\
\dot{\tilde{d}}
\end{matrix}\right] &= \hat{A}_{CL}\left[\begin{matrix}
x\\
\tilde{x}\\
\tilde{d}
\end{matrix}\right] + B_{CL}d_f \nonumber\\ 
y_{CL} &= x
\end{align}
where $d_f = \mathbf{a}x$ is the feedback disturbance to the closed-loop nominal system, which depends on the ``output'' of the closed-loop system $y_{CL}$, and 
\begin{align}
    B_{CL} \triangleq \left[\begin{matrix}
        0 &
        0 &
        1 &
        0 &
        0 &
        -1 &
        0
    \end{matrix} \right]^T
\end{align}
gives us the disturbance's effect on both the plant and observer error dynamics. Taking
\begin{align}
    C_{CL} &\triangleq \left[ \begin{matrix}
        I_3 & \mathbf{0}_{3\times 4}
    \end{matrix} \right], 
\end{align}
we can write the transfer function of the closed-loop system, with the disturbance $d_f$ as the input and the plant state $x$, as the output
\begin{align}\label{eq:CLTF}
    H_{CL}(s) \triangleq C_{CL}\left(sI_7 - A_{CL} \right)B_{CL},
\end{align}
which depends only on the controller parameters and not on the unknown plant parameters. Now, a disturbance which depends on the plant state, such as $d_f$ in this example, can be written as being in a feedback connection with this transfer function. 
Designing $K$ and $G$ such that~\eqref{eq:CLTF} has small $\mathcal{L}_2$ gain could then help to guarantee stability for a wider range of plants, using results such as the small-gain theorem~\cite{khalil2002nonlinear,ogata2010modern} for example. 

\section{Conclusions}
This work shows that ADRC can be used to stabilize a class of 3rd order disturbed linear plants for desired performance, without requiring the plant parameters to be known to the observer, and without requiring arbitrarily high observer gains in general, by showing how the gains can be chosen to achieve desired closed-loop eigenvalues. It further shows that stability is possible for stable disturbances, and conjectures that the main result could be extended to arbitrarily large plants, with equal order and relative degree. 
Additionally, it shows how ADRC can recover the performance of model-based observers, if its gains are chosen properly. 
Because providing a way to find the desired ADRC gains, without knowing the plant parameters, is beyond the scope of this work, it instead points to promising directions for future work to parameterize ADRC in novel ways. 

\section*{Acknowledgment}

This work was supported in part by the Department of the Navy, Office of Naval Research (ONR), under federal grant N00014-22-1-2207. 

\appendix
\subsection{Proof of Remark~\ref{rm:PIDisInsufficient}}

Here, we show that PID control is insufficient to stabilize our third-order linear system of relative degree $3$, without additional assumptions. Starting from our canonical form of the system~\eqref{eq:canonical3Dplant}, we assume for simplicity that $d(t)=0$, for $t \geq 0$. PID control takes the form
\begin{align*}
u &= -k_Px_1 - k_Dx_2 - k_I\hat{d},
\end{align*}
where we've assumed that the controller has access to both the output $x_1$ and its derivative $x_2$, and where the controller has virtual state
$\dot{\hat{d}} = x_1 $.
This makes the closed-loop system, under PID control, 
\begin{align}\label{eq:PID3dCL}
\left[\begin{matrix}
\dot{\hat{d}}\\
\dot{x}
\end{matrix}\right] &= \left[ \begin{matrix}
0 & 1 & 0 & 0\\
0 & 0 & 1 & 0\\
0 & 0 & 0 & 1\\
-bk_I & a_1-bk_P & a_2-bk_D & a_3
\end{matrix} \right]\left[\begin{matrix}
\hat{d}\\
x
\end{matrix}\right].
\end{align}
The trace of a matrix is the sum of its diagonal elements and it is equal to the sum of all the eigenvalues. Note that the trace of the system matrix of~\eqref{eq:PID3dCL} is $a_3$, indicating that it depends only on a plant parameter and not on any control parameters. The system~\eqref{eq:PID3dCL} is stable if and only if its eigenvalues each have negative real part, and a necessary condition for that is that the trace of the system matrix is negative. Therefore, a necessary condition for PID control to stabilize this system is $a_3 < 0$, indicating that PID control cannot stabilize this system in general. 

{\scriptsize

    \bibliographystyle{ieeetr}

  }

\end{document}